\documentclass[final]{siamltex}
\usepackage{framed,algorithm,algorithmic}
\usepackage{amsfonts,amsmath,bm}
\long\def\symbolfootnote[#1]#2{\begingroup%
\def\thefootnote{\fnsymbol{footnote}}\footnote[#1]{#2}\endgroup}

\newcommand{\Exp}{\mbox{}{\mathbb{E}}}
\newcommand{\Expect}[1]{\mbox{}{\mathbb{E}}\left[#1\right]}

\newcommand{\FNorm }[1]{\mbox{}\|#1\|_\mathrm{F}  }
\newcommand{\FNormS}[1]{\mbox{}\|#1\|_\mathrm{F}^2}

\newcommand{\TNorm }[1]{\mbox{}\|#1\|_2  }
\newcommand{\TNormS}[1]{\mbox{}\|#1\|_2^2}

\newcommand{\XNorm }[1]{\mbox{}\|#1\|_{\xi}  }
\newcommand{\XNormS}[1]{\mbox{}\|#1\|_{\xi}^2}

\newcommand{\transp}{^{\textsc{T}}}
\newcommand{\trace}{\text{\rm Tr}}
\newcommand{\mat}[1]{{\ensuremath{\bm{\mathrm{#1}}}}}

\def\rank{\hbox{\rm rank}}

\def\b{{\mathbf b}}
\def\e{{\mathbf e}}

\def\s{{\mathbf s}}
\def\u{{\mathbf u}}
\def\v{{\mathbf v}}

\def\matA{\mat{A}}
\def\matB{\mat{B}}
\def\matC{\mat{C}}

\def\matE{\mat{E}}

\def\matI{\mat{I}}

\def\matP{\mat{P}}
\def\matQ{\mat{Q}}
\def\matR{\mat{R}}
\def\matS{\mat{S}}

\def\matU{\mat{U}}
\def\matV{\mat{V}}
\def\matW{\mat{W}}
\def\matX{\mat{X}}
\def\matY{\mat{Y}}
\def\matZ{\mat{Z}}
\def\matOmega{\mat{\Omega}}
\def\matSig{\mat{\Sigma}}

\def\MatPsi{\mat{\Psi}}

\def\scl{{\textsc{l}}}
\def\scu{{\textsc{u}}}
\def\phiu{{\overline{\phi}}}

\def\phil{{\underbar{\math{\phi}}}}
\DeclareMathSymbol{\Prob}{\mathbin}{AMSb}{"50}
\newcommand\remove[1]{}

\def\math#1{$#1$}

\def\mand#1{$$#1$$}
\def\frac#1#2{{#1\over #2}}

\def\mld#1{\begin{equation}
#1
\end{equation}}
\def\eqar#1{\begin{eqnarray}
#1
\end{eqnarray}}
\def\eqan#1{\begin{eqnarray*}
#1
\end{eqnarray*}}

\DeclareMathSymbol{\R}{\mathbin}{AMSb}{"52}

\def\cl#1{{\cal #1}}

\def\argmin{\mathop{\hbox{argmin}}\limits}

\def\x{{\mathbf x}}

\def\z{{\mathbf z}}

\def\a{{\mathbf a}}
\def\b{{\mathbf b}}

\def\norm#1{{\left\|#1\right\|}}
\def\mynorm#1{{\|#1\|}}

\def\ceil#1{{\left\lceil\,#1\,\right\rceil}}
\def\r#1{{(\ref{#1})}}

\begin{document} 

\title{Near-optimal Column-based Matrix Reconstruction}

\author{
Christos Boutsidis\thanks{Mathematical Sciences Department, IBM T. J. Watson Research Center, Yorktown Heights, NY 10598. Email: \texttt{cboutsi@us.ibm.com}}
\and
Petros Drineas\thanks{Computer Science Department, Rensselaer Polytechnic Institute, Troy NY, 12180. Email: \texttt{\{drinep,magdon\}@cs.rpi.edu}}
\and
Malik Magdon-Ismail\footnotemark[2]
}

\maketitle

\begin{abstract}
\noindent We consider low-rank reconstruction of a matrix using a subset of its columns and we present asymptotically optimal algorithms for both spectral norm and Frobenius norm reconstruction. The main tools we introduce to obtain our results are: (i) the use of fast approximate SVD-like decompositions for column-based matrix reconstruction, and (ii) two deterministic algorithms for selecting rows from matrices with orthonormal columns, building upon the sparse representation theorem for decompositions of the identity that appeared in~\cite{BSS09}.
\end{abstract} 



\pagestyle{myheadings}
\thispagestyle{plain}
\markboth{BOUTSIDIS, DRINEAS, AND MAGDON-ISMAIL}{NEAR-OPTIMAL COLUMN-BASED MATRIX RECONSTRUCTION}

\section{Introduction}

The best rank \math{k} approximation to a matrix
\math{\matA \in \R^{m \times n}} is
$$\matA_k=\sum_{i=1}^k\sigma_i\u_i\v_i\transp,$$ where
\math{\sigma_1\ge\sigma_2\ge\cdots\ge\sigma_k\ge 0} are the top \math{k}
 singular values of $\matA$, with associated left and right singular vectors
\math{\u_i \in \R^m} and~\math{\v_i \in \R^n} respectively. (See Section~\ref{sxn:notation} for notation and recall that the singular values and singular vectors of $\matA$ can be computed via the Singular Value Decomposition (SVD) of $\matA$ in \math{O( mn\min\{m,n\})} time.) It is well-known that $\matA_k$ optimally approximates $\matA$ among all rank $k$ matrices, with respect to any unitarily invariant norm.  There is considerable interest (e.g. \cite{CH92,DR10,DV06,DRVW06,DKR02,FKV98,HMT,LWMRT07,MD09}) in determining a minimum set of $r \ll n$ columns of $\matA$  which is approximately as good as \math{\matA_k} at reconstructing \math{\matA}.
Such columns are important for interpreting data \cite{MD09}, building robust machine learning algorithms \cite{CH92}, feature selection, etc.

Let $\matA \in \mathbb{R}^{m \times n}$ and let $\matC \in \mathbb{R}^{m \times r}$ consist of $r$ columns of $\matA$ for some $k \leq r < n$.
We are interested in the reconstruction errors
$$
\mynorm{\matA-\matC\matC^+\matA}_{\xi} \qquad \mbox{and} \qquad
\mynorm{\matA-\Pi^\xi_{\matC,k}(\matA)}_{\xi},
$$
for $\xi=2,\mathrm{F}$ (see Section~\ref{sxn:notation} for notation).
The former is the reconstruction error
for \math{\matA} using the columns in \math{\matC};
the latter is the error from the
best  rank \math{k}
reconstruction of
\math{\matA} (under the appropriate norm) within the column space of $\matC$.
For fixed $\matA$, $k$, and $r$, we would like these errors to be as close
to $$\XNorm{\matA-\matA_k}$$ as possible.
We present polynomial-time near-optimal constructions for arbitrary
$r>k$,
settling important open questions regarding column-based matrix reconstruction.
\begin{itemize}
\item {\bf Spectral norm:}
What is the best reconstruction error with
\math{r > k} columns? We present polynomial-time (deterministic and
randomized) algorithms with approximation error asymptotically matching a lower
bound proven in this work. Prior work had focused on the $r=k$ case and presented near-optimal polynomial-time algorithms~\cite{ DR10,GE96}.
\item{\bf Frobenius norm:} How many columns are needed for relative error
approximation, i.e. a reconstruction error of at most
$$\left(1+\epsilon \right)\norm{\matA-\matA_k}_{\mathrm{F}},$$
for $\epsilon>0$? We show that
$O(k/\epsilon)$
columns contain a rank-\math{k} subspace which
reconstructs \math{\matA} to relative error,
and we present the first  sub-SVD (in terms of running time) randomized algorithm to identify these columns. This matches
the \math{\Omega(k/\epsilon)} lower bound in~\cite{DV06} and improves the
best known upper bound of \math{O(k\log k+ k/\epsilon)}~\cite{DR10,DV06,DMM06b,Sar06}.
%
\end{itemize}

\subsection{Notation}\label{sxn:notation}

\math{\matA,\matB,\ldots} are matrices; \math{\a,\b,\ldots} are
column vectors. $\matI_{n}$ is the $n \times n$
identity matrix;  $\bm{0}_{m \times n}$ is the $m \times n$ matrix of zeros; $\bm{1}_n$ is the $n \times 1$ vector of ones; $\bm{e}_i$ is the standard basis (whose dimensionality will be clear from the context); $\rank(\matA)$ is the rank of $\matA$.
The Frobenius and the spectral matrix-norms are:
$ \FNormS{\matA} = \sum_{i,j} \matA_{ij}^2$
and $\TNorm{\matA} = \max_{\TNorm{\x}=1}\TNorm{\matA\x}$;
 $\XNorm{\matA}$ is used if a result holds for
both norms $\xi = 2$ and $\xi = \mathrm{F}$.
The Singular Value Decomposition (SVD) of
$\matA$, with $\rank(\matA) = \rho$ is
%
\begin{eqnarray*}
\label{svdA} \matA
         = \underbrace{\left(\begin{array}{cc}
             \matU_{k} & \matU_{\rho-k}
          \end{array}
    \right)}_{\matU_{\matA} \in \R^{m \times \rho}}
    \underbrace{\left(\begin{array}{cc}
             \matSig_{k} & \bf{0}\\
             \bf{0} & \matSig_{\rho - k}
          \end{array}
    \right)}_{\matSig_\matA \in \R^{\rho \times \rho}}
    \underbrace{\left(\begin{array}{c}
             \matV_{k}\transp\\
             \matV_{\rho-k}\transp
          \end{array}
    \right)}_{\matV_\matA\transp \in \R^{\rho \times n}},
\end{eqnarray*}
with singular values \math{\sigma_1\ge\ldots\sigma_k\geq\sigma_{k+1}\ge\ldots\ge\sigma_\rho > 0}.
We will use $\sigma_i\left(\matA\right)$ to denote the $i$-th
singular value of $\matA$ when the matrix is not clear from the context.
The matrices
$\matU_k \in \R^{m \times k}$ and $\matU_{\rho-k} \in \R^{m \times (\rho-k)}$ contain the left singular vectors of~$\matA$, and, similarly, the matrices $\matV_k \in \R^{n \times k}$ and $\matV_{\rho-k} \in \R^{n \times (\rho-k)}$ contain the right singular vectors of~$\matA$. It is well-known that $\matA_k=\matU_k \matSig_k \matV_k\transp$ minimizes \math{\XNorm{\matA - \matX}} over all
matrices \math{\matX \in \R^{m \times n}} of rank at most $k$. We use $\matA_{\rho-k}$ to denote the matrix $\matA - \matA_k = \matU_{\rho-k}\matSig_{\rho-k}\matV_{\rho-k}\transp$. Also,
$\matA^+ = \matV_\matA \matSig_\matA^{-1} \matU_\matA\transp$ denotes the Moore-Penrose pseudo-inverse of $\matA$.
For a symmetric positive definite matrix $\matA=\matB\matB\transp$, $\lambda_{i}\left(\matA\right) = \sigma_{i}^2\left(\matB\right)$ denotes the $i$-th eigenvalue of $\matA$.

Finally, given a matrix $\matA \in \mathbb{R}^{m \times n}$ and a matrix $\matC \in \mathbb{R}^{m \times r}$ with $r > k$, we formally define the matrix $\Pi_{\matC,k}^{\xi}(\matA)\in \mathbb{R}^{m \times n}$ as the best approximation to $\matA$ within the column space of $\matC$ that has rank at most
$k$.  $\Pi_{\matC,k}^{\xi}(\matA)$
minimizes the residual
$\mynorm{\matA-\hat\matA}_\xi,$
over all
\math{\hat\matA} in the column space of
\math{\matC} that have rank at most \math{k} (one can write
\math{\Pi_{\matC,k}^{\xi}(\matA)=\matC\matX} where
\math{\matX\in \mathbb{R}^{r \times n}} has rank at most $k$).
In general, $$\Pi_{\matC,k}^{2}(\matA) \neq \Pi_{\matC,k}^{\mathrm{F}}(\matA);$$ Section~\ref{sec:bestrankk} discusses the computation of~$\Pi_{\matC,k}^{\xi}(\matA)$.


\subsection{Main results}\label{sec:main0}

Since
$$\XNorm{\matA-\matC\matC^+\matA} \le \XNorm{\matA-\Pi^{\xi}_{\matC,k}(\matA)},$$
we will state all our bounds in terms of the latter quantity. Note that we chose to state our Frobenius norm bounds in terms of the square of the Frobenius norm; this choice facilitates comparisons with prior work and simplifies our proofs (see also Table~\ref{table:results} for a summary of our results).
\begin{theorem}[Deterministic spectral norm reconstruction]
\label{theorem:intro1}
Given \math{\matA\in\R^{m \times n}} of rank \math{\rho}
and a target rank \math{k < \rho},
there exists a deterministic polynomial-time algorithm to
select \math{r > k} columns of \math{\matA} and form a matrix
$\matC\in\R^{m \times r}$ such that
\eqan{
\TNorm{\matA - \Pi^2_{\matC,k}(\matA)}
\leq
\textstyle\left(1 + \frac{1 + \sqrt{ (\rho-k)/r }}{ 1 - \sqrt{ k/r }  }  \right)
\TNorm{\matA - \matA_k}
= O\left(\sqrt{{\rho}/{r}}\right)
\TNorm{\matA - \matA_k}.
}
The matrix $\matC$ can be computed in
\math{T_{SVD} + O\left(rn\left(k^2+\left(\rho-k\right)^2\right)\right)} time, where $T_{SVD}$ is the time needed to compute all $\rho$ right singular vectors of $A$.
\end{theorem}

Our algorithm uses
 the matrices $\matV_k$ and $\matV_{\rho-k}$ of the right singular vectors of $\matA$. These matrices can be computed in $O(mn \min\{m,n\})$ time via the SVD. The asymptotic multiplicative error of the above theorem matches a lower bound that we prove in Section~\ref{sec:lower}. This is the first spectral reconstruction algorithm with asymptotically optimal guarantees for arbitrary $r > k$. Previous work presented near-optimal algorithms for $r=k$ \cite{GE96}. We note that in Section~\ref{sec:PROOFS} we will present a result that achieves a slightly worse error bound (essentially replacing $\rho$ by $n$ in the accuracy guarantee), but only uses the top $k$ right singular vectors of $\matA$ (i.e., the matrix $\matV_k$ from the SVD of $\matA$).

\begin{theorem}[Deterministic Frobenius norm reconstruction]
\label{theorem:intro2}
Given \math{\matA\in\R^{m\times n}} of rank $\rho$ and a target rank $k <\rho$,
there exists a deterministic
polynomial-time algorithm to
select \math{r>k} columns of \math{\matA} and form a matrix
$\matC\in\R^{m \times r}$
such that
$$ \FNormS{\matA - \Pi_{\matC,k}^{\mathrm{F}}(\matA)} \leq \textstyle
\left(1 +  \left(1 - \sqrt{ k/r }\right)^{-2}  \right)
\FNormS{\matA - \matA_k}.$$
The matrix $\matC$ can be computed in
\math{T_{\matV_k} + O\left(mn + nrk^2\right)} time, where $T_{\matV_k}$ is the time needed to compute the top $k$ right singular vectors of $\matA$.
\end{theorem}

Our bound implies a constant-factor approximation.
Previous work presents deterministic near-optimal algorithms for $r=k$ \cite{DR10}; we are unaware of any deterministic algorithms for $r>k$.

The next two theorems guarantee (up to small constant factors) the same bounds as Theorems~\ref{theorem:intro1} and~\ref{theorem:intro2},
but the proposed algorithms are considerably more efficient. In particular, there is no need to exactly compute the right singular vectors of $\matA$,
because approximations suffice.

\begin{theorem}[Fast spectral norm reconstruction]
\label{thmFast1}
Given \math{\matA\in\R^{m\times n}} of rank $\rho$, a target rank $2\leq k < \rho$, and $0 < \epsilon < 1$,
there exists a randomized algorithm to select \math{r > k} columns of \math{\matA} and form a matrix
$\matC\in\R^{m \times r}$ such that
\eqan{
\Expect{ \TNorm{\matA - \Pi_{\matC,k}^2(\matA)} }
\leq
\textstyle\left(\sqrt{2}+\epsilon\right) \left( \frac{1 + \sqrt{ n/r }}{1 - \sqrt{ k/r }}\right)
\TNorm{\matA-\matA_k}
=
O\left(\sqrt{{n}/{r}}\right)\TNorm{\matA-\matA_k}.
}
The matrix $\matC$ can be computed in $O\left(mnk\epsilon^{-1}
\log\left( k^{-1}\min\{m,n\}\right)+nrk^2\right)$ time.
\end{theorem}
\begin{theorem}[Fast Frobenius norm reconstruction]
\label{thmFast2}
Given \math{\matA\in\R^{m\times n}} of rank $\rho$, a
target rank $2\leq k < \rho$, and  $0 < \epsilon < 1$,
there exists a randomized algorithm to
select \math{r > k} columns of \math{\matA} and form a matrix
$\matC\in\R^{m \times r}$ such that
$$
 \Expect{\FNormS{\matA - \Pi_{\matC,k}^{\mathrm{F}}(\matA)}} \leq
\textstyle
(1+\epsilon){
\left(1 + \left(1 - \sqrt{ k/r }\right)^{-2} \right)
}
\FNormS{\matA - \matA_k}.$$
The matrix $\matC$ can be computed in $O\left(mnk\epsilon^{-1}+nrk^2\right)$ time.
\end{theorem}

\begin{table}
\begin{center}
\begin{tabular}{l|c@{\hspace*{0.05in}}l|c@{\hspace*{0.05in}}l}
  & \multicolumn{2}{c|}{Spectral norm ($\xi=2$)} & \multicolumn{2}{c}{Frobenius norm ($\xi=\mathrm{F}$)}\\
&&&&\\[-9pt]
\hline
&&&&\\[-8pt]
Deterministic
&  $\left(1 + \frac{1 + \sqrt{ (\rho-k)/r }}{ 1 - \sqrt{ k/r }  } \right)^2  $ &(Thm.~\ref{theorem:intro1}) & $1 + \left(1 - \sqrt{ k/r }\right)^{-2}$ & (Thm.~\ref{theorem:intro2})\\
&&&&\\[-7pt]
\hline
&&&&\\[-8pt]
Randomized\math{^*}
& $O\left({n}/{r}\right)$  & (Thm.~\ref{thmFast1})
& $ 1+\frac{2k}{r}\bigl(1+o(1)\bigr)$ &(Thm.~\ref{thmFast3})    \\[5pt]
\hline
\end{tabular}
\medskip
\caption{\noindent Upper bounds for the approximation ratio \math{\mynorm{\matA-\Pi^\xi_{\matC,k}(\matA)}_{\xi}^2/\mynorm{\matA-\matA_k}_{\xi}^2}, for any $r >k$. Here $\matA \in \mathbb{R}^{m \times n}$ of rank $\rho$.
\math{^*} We give a bound on the expected approximation ratio.
\label{table:results}}
\end{center}
\end{table}

Our last, yet perhaps most interesting result,
guarantees relative-error Frobenius norm approximation by
combining the algorithm of Theorem~\ref{thmFast2} with one round of adaptive sampling~\cite{DV06,DRVW06}. This is the \emph{first} relative-error approximation
 for Frobenius norm reconstruction that uses a linear  number of columns
in $k$ (the target rank).  Previous work~\cite{DMM06b,Sar06, DV06, DR10}
 achieves relative error with \math{O(k\log k+ k/\epsilon)} columns.
Our result asymptotically
matches the \math{\Omega(k/\epsilon)} lower bound in \cite{DV06}.

Notice that in Theorems~\ref{theorem:intro2} and~\ref{thmFast2}, which use the deterministic spectral sparsification technique of Lemma~\ref{lemma:intro2} to select the columns,
we only achieve a $2+\epsilon$ error by selecting $O(k/\epsilon^2)$ columns. To improve this constant factor approximation to a relative error bound we used the adaptive
sampling idea from~\cite{DV06,DRVW06}.
\begin{theorem}[Fast relative-error Frobenius norm reconstruction]
\label{thmFast3}
Given \math{\matA\in\R^{m\times n}} of rank $\rho$, a target rank $2\leq k < \rho$, and $0 < \epsilon < 1$,
there exists a randomized algorithm to select at most
$$
r =\frac{2k}{\epsilon}\bigl(1+o(1)\bigr)
$$
columns of \math{\matA} and form a matrix
$\matC\in\R^{m \times r}$
such that,
\mand{\Expect{ \FNormS{\matA - \Pi_{\matC,k}^{\mathrm{F}}(\matA)}} \leq (1+\epsilon)\FNormS{\matA-\matA_k}.
}
The matrix $\matC$ can be computed in $O\left(\left(mnk+nk^3\right)\epsilon^{-2/3}\right)$ time.
\end{theorem}

\subsection{Running times} Notice that the running times in the theorems presented above are stated in terms of the number of
operations needed to compute the matrix $\matC$, and,
for simplicity, we assume that \math{\matA} is dense;
 if \math{\matA} is sparse, additional savings might be possible.
Our accuracy guarantees are in terms of the optimal matrix
$\Pi_{\matC,k}^{\xi}(\matA)$, which would require additional time to compute.
For the Frobenius norm,
computing~$\Pi_{\matC,k}^{\mathrm{F}}(\matA)$ is straightforward,
and only requires an
additional $O\left(mnr + \left(m+n\right)r^2\right)$ time (see the
discussion in Section~\ref{sec:bestrankk}).
For the spectral norm,
we are not aware of any efficient algorithm to compute
$\Pi_{\matC,k}^{2}(\matA)$ exactly.
In Section~\ref{sec:bestrankk} we present a simple approach
that computes $\hat\Pi_{\matC,k}^{2}(\matA)$,
a constant-factor approximation
to $\Pi_{\matC,k}^{2}(\matA)$, in $O\left(mnr + \left(m+n\right)r^2\right)$
time.
Our bounds in Theorems~\ref{theorem:intro1} and~\ref{thmFast1}
can be restated in terms of
the error 
$$\mynorm{\matA-\hat\Pi_{\matC,k}^{2}(\matA)}_2;$$
the accuracy guarantees only weaken by small constant factors.

\subsection{Lower Bounds}
Table \ref{table:lower} provides a summary on lower bounds for the ratio
$$\frac{\XNorm{\matA-\Pi_{\matC,k}^{\xi}(\matA)}^2}{\XNorm{\matA-\matA_k}^2},$$
where $\matC$ is a matrix consisting of $r$ columns of $\matA$, with $r \geq k$. Theorem~\ref{theorem:lower1} contributes a new lower bound for the spectral norm case when $r > k$. Note that any lower bound for the ratio \math{\XNormS{\matA-\matC\matC^+\matA}/\XNormS{\matA-\matA_k}} implies the same lower bound for \math{\XNormS{\matA-\Pi_{\matC,k}^{\xi}(\matA)}/\XNormS{\matA-\matA_k}}; the converse, however, is not true.

\subsection{Prior results on column-based matrix reconstructions} \label{sec:prior}
There is a long literature on algorithms for column-based matrix reconstruction using $r \geq k$ columns. The first result goes back to \cite{Golub65}, with the most recent one being, to the best of our knowledge, the work in~\cite{DR10}.
\begin{table}
\begin{center}
\begin{tabular}{l|c@{\hspace*{0.05in}}l|c@{\hspace*{0.05in}}l}
\# Columns (\math{r}) &\multicolumn{2}{c|}{Spectral norm ($\xi=2$)} & \multicolumn{2}{c}{Frobenius norm ($\xi=\mathrm{F}$)}\\
&&&&\\[-9pt]
\hline
&&&&\\[-8pt]
\textbf{$r=k$}
&  $n/k$ &\cite{DR10} & $k+1$& \cite{DRVW06}   \\
&&&&\\[-7pt]
\hline
&&&&\\[-8pt]
\textbf{$r>k$}
& ${n}/{r}$& (Section \ref{sec:lower})
& $ 1 + {k}/{r}$& \cite{DV06} (and Section~\ref{sec:lowerF}) \\[5pt]
\hline
\end{tabular}
\medskip
\caption{\noindent Lower bounds for the approximation ratio \math{\mynorm{\matA-\Pi^\xi_{\matC,k}(\matA)}_{\xi}^2/\mynorm{\matA-\matA_k}_{\xi}^2}. Here $\matA \in \mathbb{R}^{m \times n}$.
\label{table:lower}}
\end{center}
\end{table}

\subsubsection{The Frobenius norm case}

We present known upper bounds for the approximation ratio
$$\FNormS{\matA-\Pi_{\matC,k}^{\mathrm{F}}(\matA)} / \FNormS{\matA-\matA_k}.$$
We start with the $r = k$ case.~\cite{BMD09a} describes a $T_{\matV_k} + O\left(  nk + k^3 \left(\log^2 k\right) \left(\log \log k\right) \right)$
time randomized algorithm which provides an upper bound $O\left( k \log^{\frac{1}{2} } k \right)$ with constant probability.
This bound was subsequently improved in~\cite{DR10}. More precisely,
Theorem 8 of~\cite{DR10} gives a $(k+1)$ deterministic approximation
running in $O(knm^3 \log m)$ time;
this upper bound matches the lower
bound in~\cite{DRVW06}.
\cite{DR10} also presents three randomized algorithms such that $ \Expect{ \FNormS{\matA-\Pi_{\matC,k}^\mathrm{F}(\matA)} } = (k+1)\FNormS{\matA-\matA_k}$.
These randomized algorithms are presented in Theorem 7, Proposition 16, and Proposition 18 and run in
$O\left( knm^3 \log m \right)$, $O\left( kn^3m + k n ^4 \log n \right)$, and $O\left( k T_{SVD} +   knm^2 \right)$ time, respectively.
Moreover, Theorem 9 in~\cite{DR10} presents a randomized algorithm that runs in time 
$O\left( m n \log n k^2 \epsilon^{-2} + n \log^3 n \cdot k^7 \epsilon^{-6} \log\left(k \epsilon^{-1} \log n\right) \right)$
such that, with constant probability, 
$ \FNormS{\matA-\Pi_{\matC,k}^\mathrm{F}(\matA)} \le (1+\epsilon) \cdot (k+1)\FNormS{\matA-\matA_k}$,
for any $0 < \epsilon < 1$. Finally,~\cite{GS2011} improved upon the running time of the results in~\cite{DR10}.
More precisely, Theorem 2 in~\cite{GS2011} gives an $O\left( knm^2 \right)$ time randomized algorithm with
a $(k+1)$ multiplicative error (in expectation).

When $r=\Omega(k\log k)$, relative-error approximations are known.
\cite{DMM06b} presented the first result that achieved such a bound, using random sampling of the columns of $\matA$ according to the Euclidean norms of the rows of $\matV_k$, the so-called leverage scores~\cite{MD09}. More specifically, a \math{(1+\epsilon)}-approximation was proven using \math{r = \Omega\left(k\epsilon^{-2}\log \left(k\epsilon^{-1}\right)\right)} columns in $T_{\matV_k} + O(kn+r\log r)$ time. \cite{Sar06} argued
that the same technique gives a \math{(1+\epsilon)}-approximation
using \math{r = \Omega\left(k\log k + k \epsilon^{-1}\right)} columns. It also  showed how to improve the running time to $T_{\tilde{\matV}_k} +
O(kn+r\log r)$, where $\tilde{\matV}_k \in \R^{n \times k}$ contains the right singular vectors of an approximation to $\matA_k$ and can be computed in
\math{o(mn\min\{m,n\})} time, which is less than the time needed to compute the SVD of $\matA$.
In~\cite{DV06}, the authors leveraged volume sampling and presented an approach that achieves a relative error approximation using \math{O(k^2\log k + k\epsilon^{-1})} columns in \math{O(mnk^2\log k)} time. Also, it is possible to combine the fast volume sampling approach in~\cite{DR10} (setting, for example, $\epsilon=1/2$) with \math{O(\log k)} rounds of adaptive sampling as described in~\cite{DV06} to achieve a relative error approximation using \math{O\left(k\log k +k\epsilon^{-1}\right)} columns. The running time of this combined algorithm is $O\left(mnk^2 \log n + nk^7 \log^{3} n\cdot \log\left(k \log n\right)\right)$.
The techniques in \cite{DMM06b} do not apply to general \math{r>k}, since \math{\Omega(k\log k )} columns must be sampled in order to preserve rank with random sampling.

A related line of work (including~\cite{DV07,FL11, FMSW10,SV07})
has focused on the construction of coresets and
sketches for high dimensional subspace approximation with respect
to general $\ell_p$ norms. In our setting, \math{p=2} corresponds to
Frobenius norm matrix reconstruction, and Theorem 1.3 of~\cite{SV07}
presents an exponential in $k/\epsilon$ algorithm to select
 $O\left(k^2\epsilon^{-1} \log\left(k/\epsilon\right)\right)$ columns that
guarantee a relative error approximation. It would be interesting to
understand if the techniques of~\cite{DV07,FL11, FMSW10,SV07}
can be extended to match our results here in the special case
of \math{p=2}.

The recent work in~\cite{GS2011} presents a deterministic and a randomized algorithm for arbitrary
$r \ge k$ that guarantee upper bounds for the ratio $ \FNormS{\matA - \matC \matC^+ \matA} / \FNormS{\matA - \matA_k}$.
More precisely, Theorem 1 in~\cite{GS2011} presents an $O\left(rnm^3 \log m \right)$ time deterministic algorithm with bound $(r+1)/(r+1-k)$, which is tight up to
low order terms if $r = o(n)$. Also, Theorem 2 in~\cite{GS2011} presents an $O(rnm^2)$ time randomized algorithm
which achieves the same bound in expectation. We should notice that it is not obvious how to
extend the results in~\cite{GS2011} to obtain comparable bounds for the ratio
$\FNormS{\matA-\Pi_{\matC,k}^{\mathrm{F}}(\matA)} / \FNormS{\matA-\matA_k}.$

\subsubsection{The spectral norm case}

We present known guarantees for the approximation ratio
$$  \TNormS{\matA-\Pi_{\matC,k}^2(\matA)} / \TNormS{\matA-\matA_k} .$$
In general,
results for spectral norm have been rare.
When \math{r=k},
the strongest bound emerges from the Strong Rank Revealing
 QR (RRQR) \cite{GE96} (specifically Algorithm 4 in \cite{GE96}), which, for
\math{f>1}, runs in $O(mnk \log_{f} n )$ time and guarantees an
\math{ f^2 k(n-k) +1 } approximation. For \math{r>k}, to the
best of our knowledge, there is no easy way to extend the
RRQR guarantees. In fact we are only aware of one bound that is
applicable to this domain, other than those obtained by trivially
extending the Frobenius norm bounds, because
any $\alpha$-approximation in the Frobenius norm gives an
$\alpha(\rho-k)$-approximation in the spectral norm:
\begin{eqnarray*}
\TNorm{\matA -\Pi_{\matC,k}^2(\matA)}^2 &\leq&
\TNorm{\matA -\Pi_{\matC,k}^{\mathrm{F}}(\matA)}^2 \leq
\FNorm{\matA -\Pi_{\matC,k}^{\mathrm{F}}(\matA)}^2\\
&\leq&
\alpha \FNorm{\matA-\matA_k}^2 \leq \alpha(\rho-k)\TNorm{\matA-\matA_k}^2.
\end{eqnarray*}
Finally, recent work~\cite{AB11} describes a deterministic $T_{\matV_k}+O\left(nk\left(n-r\right)\right)$ time algorithm that guarantees an approximation error 
$$ \TNormS{\matA-\Pi_{\matC,k}^2(\matA)} / \TNormS{\matA-\matA_k} \le 2 + k(n-r) / (r-k+1),$$ for any $r \ge k$.

\section{Matrix norm properties and the computation of $\Pi_{\matC,k}^\xi(\matA)$}
\label{sec:prel}

\subsection{Matrix norm properties}

Recall notation from Section~\ref{sxn:notation}; for any matrix $\matA$ of rank at most $\rho$, it is well-known that $ \FNormS{\matA} = \sum_{i=1}^\rho\sigma_i^2(\matA) $ and $\TNorm{\matA} = \sigma_1(\matA)$. Also, the best rank $k$ approximation to $\matA$ satisfies $\TNorm{\matA-\matA_k} = \sigma_{k+1}(\matA)$ and $\FNormS{\matA-\matA_k} = \sum_{i=k+1}^{\rho}\sigma_{i}^2(\matA)$. For any two matrices $\matA$ and $\matB$ of appropriate dimensions, \math{\TNorm{\matA}\le\FNorm{\matA}\le\sqrt{\rho}\TNorm{\matA}}, $\FNorm{\matA\matB} \leq \FNorm{\matA}\TNorm{\matB}$, and $ \FNorm{\matA\matB} \leq \TNorm{\matA} \FNorm{\matB}$. The latter two properties are stronger versions of the standard submultiplicativity property. We refer to the next lemma as matrix-Pythogoras:
\begin{lemma}\label{lem:pyth}
If \math{\matX,\matY\in\R^{m\times n}} and
\math{\matX\matY\transp=\bm{0}_{m \times m}} or \math{\matX\transp\matY=\bm{0}_{n \times n}}, then
\eqan{
&\FNorm{\matX+\matY}^2 = \FNorm{\matX}^2+\FNorm{\matY}^2,& \\
&\max\{\TNorm{\matX}^2,\TNorm{\matY}^2\}\le
\TNorm{\matX+\matY}^2 \le \TNorm{\matX}^2+\TNorm{\matY}^2.&
}
\end{lemma}
\begin{proof}
Suppose \math{\matX\matY\transp=\bm{0}_{m \times m}}.
Then, \math{(\matX+\matY)(\matX+\matY)\transp=\matX\matX\transp + \matY\matY\transp}.
For $\xi = \mathrm{F}$,
$$
\FNormS{\matX + \matY}
= \trace\left((\matX+\matY)(\matX+\matY)\transp\right)
= \trace\left(\matX\matX\transp + \matY\matY\transp\right) = \FNormS{\matX} + \FNormS{\matY}.
$$
For $\xi=2$, let $\z$ be any vector in $\mathbb{R}^m$. Then,
$$
\TNormS{\matX+\matY}
= \max_{\TNorm{\z}=1}\z\transp (\matX+\matY)(\matX+\matY)\transp \z = \max_{\TNorm{\z}=1}\left( \z\transp\matX\matX\transp \z +
\z\transp\matY\matY\transp \z\right).
$$
We have that
$ \max_{\TNorm{\z}=1}\left( \z\transp\matX\matX\transp \z + \z\transp\matY\matY\transp \z\right)$ is at most
$$
\max_{\TNorm{\z}=1}\z\transp\matX\matX\transp \z +\max_{\TNorm{\z}=1}\z\transp\matY\matY\transp \z
= \TNormS{\matX} + \TNormS{\matY}
$$
and that
$$
\max_{\TNorm{\z}=1}( \z\transp\matX\matX\transp \z + \z\transp\matY\matY\transp \z)
\geq
\max_{\TNorm{\z}=1} \z\transp\matX\matX\transp \z
=
\TNormS{\matX},$$
since $\z\transp\matY\matY\transp \z$ is non-negative for any vector $\z$.
We get the same lower bound with \math{\TNormS{\matY}} instead, which means
we can lower bound by \math{\max\{\TNormS{\matX},\TNormS{\matY}\}}.
The case when \math{\matX\transp\matY=\bm0_{n\times n}} can be proven similarly.
\end{proof}

A projection operator \math{\matP} equals
its own square, \math{\matP^2=\matP}.
Projection operators play an important role in our analysis.
 The following lemma is well known
for nontrivial
symmetric projection matrices,
but also holds for non-symmetric (oblique) projection matrices.
\begin{lemma}[\cite{Szy06}]\label{lemma:oblique}
Let \math{\matP} be a non-null projection. Then,
\math{\mynorm{\matI-\matP}_2\le\mynorm{\matP}_2}.
\end{lemma}
(If in addition \math{\matI-\matP}, also a projection,
is non-null, then we get equality in the
above lemma.)

\subsection{Computing the best rank \math{k} approximation
$\Pi_{\matC,k}^\xi(\matA)$} \label{sec:bestrankk}

Let $\matA \in \mathbb{R}^{m \times n}$, let $k < n$ be an integer, and let $\matC \in \mathbb{R}^{m \times r}$ with $r > k$. Recall that $\Pi_{\matC,k}^\xi(\matA) \in \mathbb{R}^{m \times n}$ is the best rank \math{k} approximation to \math{\matA} in the column space of \math{\matC}. We can write $\Pi_{\matC,k}^\xi(\matA) = \matC\matX^\xi$, where
$$
\matX^\xi = \argmin_{\MatPsi \in {\R}^{r \times n}:\rank(\MatPsi)\leq k}\XNormS{\matA-
\matC\MatPsi}.
$$
In order to compute (or approximate) $\Pi_{\matC,k}^{\xi}(\matA)$ given $\matA$,
$\matC$, and $k$, we will use the following algorithm:
\begin{center}
\begin{algorithmic}[1]
\STATE Orthonormalize the columns of $\matC$ in $O(m r^2)$ time to construct the matrix $\matQ \in \R^{m \times r}$.
\STATE Compute
 $(\matQ\transp \matA)_k \in \R^{r \times n}$ via SVD
in \math{O(mnr+ nr^2)} -- the best rank-$k$ approximation of
\math{\matQ\transp\matA}.
\STATE Return 
$\matQ(\matQ\transp \matA)_k \in \mathbb{R}^{m \times n}$ in $O(mnk)$ time.
\end{algorithmic}
\end{center}
\medskip
Clearly, $\matQ(\matQ\transp \matA)_k$ is a rank $k$ matrix that lies in the column span of $\matC$. Note that though  $\Pi_{\matC,k}^{\xi}(\matA)$ can depend on \math{\xi}, our algorithm computes the same matrix, independent of \math{\xi}. The next lemma, which is essentially Lemma 4.3 in~\cite{CW09} combined with an improvment of Theorem 9.3 in~\cite{HMT},
proves that this algorithm computes $\Pi_{\matC,k}^{\mathrm{F}}(\matA)$ and a constant factor approximation to $\Pi_{\matC,k}^{2}(\matA)$.
\begin{lemma}\label{lem:bestF}
Given $\matA \in {\R}^{m \times n}$, $\matC\in\R^{m\times r}$ and an integer $k$,  the matrix $\matQ(\matQ\transp \matA)_k \in \mathbb{R}^{m \times n}$ described above (where \math{\matQ} is an orthonormal basis for the columns of \math{\matC}) can be computed in $O\left(mnr + (m+n)r^2\right)$ time and satisfies:
\begin{eqnarray*}
\mynorm{\matA-\matQ(\matQ\transp \matA)_k}_{\mathrm{F}}^2 &=& \FNormS{\matA-\Pi_{\matC,k}^{\mathrm{F}}(\matA)},\\[3pt]
\mynorm{\matA-\matQ(\matQ\transp \matA)_k}_2^2 &\leq& 2\TNormS{\matA-\Pi_{\matC,k}^{2}(\matA)}.
\end{eqnarray*}
\end{lemma}
\begin{proof}
Our proof for the Frobenius norm case is a mild modification of the proof of Lemma 4.3~\cite{CW09}. First, note that $\Pi^{\mathrm{F}}_{\matC,k}(\matA) = \Pi^{\mathrm{F}}_{\matQ,k}(\matA)$, because $\matQ \in \R^{m \times r}$ is an orthonormal basis for the column space of~$\matC$. Thus,
$$\FNormS{\matA-\Pi^{\mathrm{F}}_{\matC,k}(\matA)} = \FNormS{\matA-\Pi^{\mathrm{\mathrm{F}}}_{\matQ,k}(\matA)} = \min_{\MatPsi:\rank(\MatPsi)\leq k}\FNormS{\matA-\matQ\MatPsi}.$$
Now, using matrix-Pythagoras and the orthonormality of $\matQ$,
\mand{
\FNormS{\matA-\matQ\MatPsi}=\FNormS{\matA-\matQ\matQ\transp\matA+\matQ(\matQ\transp\matA-\MatPsi)} = \FNormS{\matA-\matQ\matQ\transp\matA}
+\FNormS{\matQ\transp\matA-\MatPsi}.
}
Setting $\MatPsi = (\matQ\transp\matA)_k$ minimizes the above quantity over all rank-$k$ matrices $\MatPsi$. Thus, combining the above results,
$\FNormS{\matA-\Pi^{\mathrm{F}}_{\matC,k}(\matA)} = \FNormS{\matA - \matQ(\matQ\transp \matA)_k}$.

We proceed to the spectral-norm part of the proof, which combines ideas from Theorem 9.3~\cite{HMT} and matrix-Pythagoras to manipulate the term $\TNormS{\matA-\matQ(\matQ\transp\matA)_k}$. In the derivations below there are two sources of errors: the first comes from projecting $\matA$ on $\matQ \matQ\transp$ and the second from taking a low-rank approximation of this projection:
\eqan{ \TNormS{\matA-\matQ(\matQ\transp\matA)_k} &=&
\TNormS{\matA-\matQ\matQ\transp\matA+
\matQ(\matQ\transp\matA-(\matQ\transp\matA)_k)}\\
&\le&
\TNormS{\matA-\matQ\matQ\transp\matA}+
\TNormS{\matQ\matQ\transp\matA-(\matQ\matQ\transp\matA)_k}\\
&\buildrel(a)\over\le&
\TNormS{\matA-\Pi_{\matQ,k}^2(\matA)}+
\TNormS{\matA-\matA_k}\\
&\le&
2\TNormS{\matA-\Pi_{\matQ,k}^2(\matA)}.
}
The first inequality follows from the simple fact that $(\matQ\matQ\transp\matA)_k=\matQ(\matQ\transp\matA)_k$ and
matrix-Pythagoras; the first term in (a) follows because \math{\matQ\matQ\transp\matA} is the (unconstrained, not necessarily of rank at most $k$) best approximation to \math{\matA} in the column space of~\math{\matQ}; the second term in (a) follows because $\matQ\matQ\transp$ is a projector matrix and thus
$$\TNormS{\matQ\matQ\transp\matA-(\matQ\matQ\transp\matA)_k}=\sigma_{k+1}^2(\matQ\matQ\transp\matA)\le\sigma_{k+1}^2(\matA) =
\TNormS{\matA-\matA_k}.$$
The last inequality follows because $\TNormS{\matA-\matA_k}\le\TNormS{\matA-\Pi_{\matQ,k}^2(\matA)}.$
\end{proof}

\section{Main Tools}\label{sec:main1}
Our two main tools are the use of matrix factorizations for
column-based low-rank matrix reconstruction, and two
deterministic sparsification lemmas which extend the work of
\cite{BSS09}.

\subsection{Matrix factorizations}\label{sec:generic}

Our first tool (see Lemmas \ref{lem:genericNoSVD}, \ref{tropp1}, and \ref{tropp2}) connects matrix factorizations and matrix reconstruction from its columns.
Specifically, Lemmas \ref{tropp1} and \ref{tropp2} consider factorizations of
the matrix \math{\matA\in\R^{m\times n}} of the form
$$\matA=\matB\matZ\transp+\matE,$$
where
\math{\matB=\matA\matZ\in\R^{m\times k}},
\math{\matZ\in\R^{n\times k}}, \math{\matE\in\R^{m\times n}},
and \math{\matZ} has orthonormal columns.
Note that the factorization decomposes \math{\matA} to its projection
\math{\matA\matZ\matZ\transp} onto the \math{k} dimensional space
spanned by the columnd of \math{\matZ} and the
 orthogonal error \math{\matE=\matA(\matI-\matZ\matZ\transp)}, which implies
 that
\math{\matE\matZ=\bm{0}_{m\times k}}.
Lemma~\ref{lem:genericNoSVD} shows how these
factorizations are connected to column selection.
Lemma \ref{lem:genericNoSVD} is the starting point of all our
column reconstruction results.
\begin{lemma}
\label{lem:genericNoSVD}
Let $\matA = \matB \matZ\transp + \matE$, with $\matB=\matA\matZ$ and $\matZ\transp\matZ=\matI_{k}$.
Let $\matS\in\R^{n\times r}$ be any matrix such that $rank(\matZ\transp \matS) =
rank(\matZ)=k.$
Let $\matC = \matA \matS \in \R^{m \times r}$. Then,
\begin{equation}\label{eqn1}
\mynorm{\matA - \Pi^{\xi}_{\matC,k}(\matA)}_\xi^2 \leq \XNormS{\matE} +
\mynorm{\matE\matS (\matZ\transp \matS)^+}_\xi^2;
\end{equation}
and,
\begin{equation}\label{eqn2}
\XNormS{\matA - \Pi^{\xi}_{\matC,k}(\matA)}
\leq \XNormS{\matE} \cdot \TNormS{ \matS (\matZ\transp \matS)^+}.
\end{equation}
\end{lemma}
\begin{proof}
We first prove Eqn.~(\ref{eqn1}).
The optimality of $\Pi^{\xi}_{\matC,k}(\matA)$ implies that
$$\XNormS{\matA - \Pi^{\xi}_{\matC,k}(\matA)}
\leq \XNormS{\matA - \matX},$$ over all matrices $\matX \in \mathbb{R}^{m \times n}$ of rank at most $k$ in the column space of \math{\matC}.
Consider the matrix $\matX = \matC(\matZ\transp \matS)^+ \matZ\transp$ (clearly \math{\matX}
is in the column space of $\matC$ and $\rank(\matX)\le k$ because
\math{\matZ\in\R^{n\times k}}):
%
\eqar{
\XNormS{\matA - \matC(\matZ\transp \matS)^+ \matZ\transp}
&=&
\XNormS{\underbrace{\matB \matZ\transp + \matE}_{\matA} - \underbrace{(\matB \matZ\transp + \matE )\matS}_{\matC=\matA\matS=(\matB \matZ\transp+\matE)\matS} (\matZ\transp\matS)^+\matZ\transp }
\nonumber
\\
&=&
\XNormS{\matB \matZ\transp - \matB \matZ\transp \matS (\matZ\transp\matS)^+\matZ\transp
+ \matE-\matE\matS (\matZ\transp\matS)^+\matZ\transp}
\nonumber
\\
&\buildrel{(a)}\over{=}&
\XNormS{\matE-\matE\matS(\matZ\transp\matS)^+\matZ\transp}
\label{eq:proof1}\\
&\buildrel{(b)}\over{\leq}& \XNormS{\matE} + \XNormS{\matE\matS(\matZ\transp\matS)^+\matZ\transp}.\nonumber
}
\math{(a)} follows because, by assumption, \math{\rank(\matZ\transp \matS)=k}, and thus $(\matZ\transp \matS) (\matZ\transp \matS)^+=\matI_{k}$ which implies
that the first two terms cancel:
$$ \matB \matZ\transp - \matB (\matZ\transp \matS) (\matZ\transp \matS)^+ \matZ\transp = \bm{0}_{m \times n}.$$
\math{(b)} follows by matrix-Pythagoras because
$$\matE\matS(\matZ\transp\matS)^+\matZ\transp \matE\transp = \bm{0}_{m \times n},$$
(recall that $\matE = \matA(\matI-\matZ\matZ\transp)$ and so
$\matE\matZ = \bm{0}_{m \times k}$).
The lemma follows by strong submultiplicativity
because $\matZ$ has orthonormal columns, hence  $\TNorm{\matZ}=1$.

We now prove Eqn.~(\ref{eqn2}).
In the above derivation up to \r{eq:proof1},
we have shown:
\eqan{\XNormS{\matA - \Pi^{\xi}_{\matC,k}(\matA)} \le
\XNormS{\matA-\matC(\matZ\transp \matS)^+ \matZ\transp}
&=&\XNormS{\matE(\matI-\matS(\matZ\transp\matS)^+\matZ\transp)}.
}
By strong submultiplicativity, the last term is bounded by
\mand{
\XNormS{\matE}\mynorm{\matI-\matS(\matZ\transp\matS)^+\matZ\transp}_2^2
\le\XNormS{\matE}\mynorm{\matS(\matZ\transp\matS)^+\matZ\transp}_2^2
=\XNormS{\matE}\mynorm{\matS(\matZ\transp\matS)^+}_2^2
}
The first step follows because
$$\matP=\matS(\matZ\transp\matS)^+\matZ\transp,$$ is a non-null
projection because
$$\matP^2=\matS(\matZ\transp\matS)^+\matZ\transp\matS(\matZ\transp\matS)^+\matZ\transp=\matS(\matZ\transp\matS)^+\matZ\transp,$$
(where we used \math{\matZ\transp\matS(\matZ\transp\matS)^+=\matI_k}), and so
we can apply Lemma~\ref{lemma:oblique};
the second step is because \math{\matZ} is orthogonal.
\end{proof}

In this work, we view $\matC$ as a dimensionally-reduced or sampled sketch of $\matA$; $\matS$ is the dimension-reduction or sampling matrix.
In words, Lemma \ref{lem:genericNoSVD} argues that if the matrix $\matS$ preserves the rank of an approximate factorization of the original matrix $\matA$, then
the reconstruction of $\matA$ from $\matC = \matA \matS$ has an error that is essentially proportional to the error of the approximate factorization.
The importance of this lemma is that it  indicates an algorithm for matrix reconstruction using a subset of the columns of $\matA$. First,
compute \emph{any} factorization of the form
$\matA = \matB \matZ\transp + \matE$, where \math{\matB=\matA\matZ} and
\math{\mynorm{\matE}_\xi} is small.
Then, compute a sampling matrix \math{\matS}
which satisfies the rank assumption and controls
the error $\mynorm{\matE \matS (\matZ\transp\matS)^+}_\xi$.

An immediate corollary of Lemma~\ref{lem:genericNoSVD} emerges by considering the SVD of $\matA$. More specifically, consider the following factorization of $\matA$:
$$\matA = \matA\matV_k\matV_k\transp + \left(\matA-\matA_k\right),$$
where $\matV_k$ is the matrix of the top $k$ right singular vectors of $\matA$. In the parlance of Lemma~\ref{lem:genericNoSVD}, $\matZ = \matV_k$, $\matB = \matA \matV_k$, $\matE = \matA -\matA_k$, and clearly \math{\matE\matZ=\bm{0}_{m \times k}}.
\begin{lemma}
\label{lem:generic}
Let $\matS \in \mathbb{R}^{n \times r}$ be a matrix such that
$\rank(\matV_k\transp\matS) = k$. Let \math{\matC=\matA\matS}. Then,
\begin{equation}\label{eqn3}
\mynorm{\matA - \Pi_{\matC,k}^{\xi}(\matA)}_\xi^2
\leq \XNormS{\matA-\matA_k} + \mynorm{(\matA-\matA_k) \matS (\matV_k\transp \matS)^+}_\xi^2;
\end{equation}
and,
\begin{equation}\label{eqn4}
\XNormS{\matA - \Pi^{\xi}_{\matC,k}(\matA)}
\leq \XNormS{\matA-\matA_k} \cdot \TNormS{ \matS (\matV_k\transp \matS)^+}.
\end{equation}
\end{lemma}

The above lemma will be useful for designing the deterministic (spectral norm and Frobenius norm) column-reconstruction algorithms of Theorems~\ref{theorem:intro1} and~\ref{theorem:intro2}. However, computing the SVD is costly and thus we would like to design a factorization of the form $\matA = \matB \matZ\transp + \matE$ that is as good as the SVD, but can be computed in $O(mnk)$ time. The next two lemmas achieve this goal by extending
the algorithms in \cite{HMT,RST09}
(see Sections~\ref{sec:proofS} and~\ref{sec:proofF} for their proofs).
We will use these factorizations to design fast column
reconstruction algorithms in Theorems \ref{thmFast1},
\ref{thmFast2}, and~\ref{thmFast3}.
\begin{lemma}[Randomized fast spectral norm SVD]
\label{tropp1}
Given \math{\matA\in\R^{m\times n}} of rank $\rho$, a target rank $2\leq k < \rho$, and
$0 < \epsilon < 1$,
there exists an algorithm that
computes a factorization  $\matA = \matB \matZ\transp + \matE$, with $\matB = \matA \matZ$, $\matZ\transp\matZ = \matI_k$, and \math{\matE\matZ=\bm{0}_{m \times k}} such that
$$\Expect{\TNorm{\matE}} \leq \left(\sqrt{2}+\epsilon\right)
\TNorm{\matA - \matA_k}.$$
The proposed algorithm runs in
$O\left(mnk\epsilon^{-1}
\log\left( k^{-1}\min\{m,n\}\right)\right)$ time.
\end{lemma}
\begin{lemma}[Randomized fast Frobenius norm SVD]\label{tropp2}
Given \math{\matA\in\R^{m\times n}} of rank $\rho$, a target rank $2\leq k < \rho$, and $0 < \epsilon < 1$, there exists an algorithm
that  computes a factorization $\matA = \matB \matZ\transp + \matE$, with $\matB = \matA \matZ$, $\matZ\transp\matZ = \matI_k$, and \math{\matE\matZ=\bm{0}_{m \times k}}
such that
$$\Expect{\FNormS{\matE}} \leq (1+{\epsilon})
\FNormS{\matA - \matA_k}.$$
The proposed algorithm runs in $O\left(mnk\epsilon^{-1}\right)$ time.
\end{lemma}

\subsection{Sparse approximate decompositions of the identity}
\label{sec:mainsparse}
Lemmas~\ref{lem:genericNoSVD}, \ref{tropp1} and \ref{tropp2} argue that,
 in order to achieve almost optimal column-based matrix reconstruction,
 we need a sampling matrix $\matS$ that preserves the rank of $\matZ$ and controls the error $\mynorm{\matE \matS(\matZ\transp \matS)^+}_\xi$. We present algorithms to compute such a matrix $\matS$ in
Lemmas~\ref{lemma:intro1} and~\ref{lemma:intro2}. 
These lemmas were motivated by an important linear-algebraic result for a decomposition of the identity presented by Batson \textit{et al.}~\cite{BSS09}.
 It is worth emphasizing that the result of~\cite{BSS09} can not be directly
applied to the column reconstruction problem. Indeed, in our setting, it is necessary to control properties related to \textit{both} matrices
\math{\matZ} and \math{\matE =\matA - \matB \matZ\transp} \emph{simultaneously}. In the spectral-norm reconstruction case, we need to control
the singular values of the two matrices;
in the Frobenius-norm reconstruction case, we need to control
singular values and Frobenius norms of two matrices. The following two lemmas are proven in Sections~\ref{sec:proofpure1} and~\ref{sec:proofpure2}.
\begin{lemma}[Dual Set Spectral Sparsification.] \label{lemma:intro1}
Let \math{\cl V=\{\v_1,\ldots,\v_n\}} and \math{\cl U=\{\u_1,\ldots,\u_n\}}
be two equal cardinality decompositions of the identity, where
\math{\v_i\in\R^{k}} ($k < n$), \math{\u_i\in\R^\ell} ($\ell \leq n$),
$\sum_{i=1}^n\v_i\v_i\transp=\matI_{k}$, and $\sum_{i=1}^n\u_i\u_i\transp=\matI_{\ell}$.
Given an integer \math{r} with \math{k < r \le n}, there exists a set of weights
\math{s_i\ge 0} ($i=1,\ldots,n$) {at most \math{r} of which are non-zero}, such that
\eqan{
\lambda_{k}\left(\sum_{i=1}^ns_i\v_i\v_i\transp\right)
\ge
\left(1 - \sqrt{\frac{k}{r}}\right)^2,
\qquad\mbox{and}
\qquad
\lambda_{1}\left(\sum_{i=1}^ns_i\u_i\u_i\transp\right)
\le \left(1 + \sqrt{ \frac{\ell}{r} }\right)^2.
}
The weights $s_i$ can be computed deterministically in $O\left(r n \left(k^2+\ell^2\right) \right)$ time.
\end{lemma}

In matrix notation, let $\matU$ and $\matV$
be the matrices whose \textit{rows} are the vectors
$\u_i$ and $\v_i$ respectively.
We can now construct the sampling matrix
$\matS \in \mathbb{R}^{n \times r}$ as follows: for $i=1,\ldots,n$, if $s_i$ is non-zero then include $\sqrt{s_i} \e_i$ as a column of $\matS$; here $\e_i$ is the $i$-th standard basis vector
\footnote{Note that we slightly abused notation: indeed, the number of columns of $\matS$ is less than or equal to $r$, since at most $r$ of the weights are non-zero. Here, we use $r$ to also denote the actual number of non-zero weights, which is equal to the number of columns of the matrix $\matS$.}.
Using this matrix notation,
\mand{
\sum_{i=1}^ns_i\v_i\v_i\transp=\matV\transp\matS\matS\transp\matV
\qquad
\hbox{and}
\qquad
\sum_{i=1}^ns_i\u_i\u_i\transp=\matU\transp\matS\matS\transp\matU,
}
and so the above lemma guarantees that
$$
\sigma_{k}(\matV\transp\matS)
\ge
1 - \sqrt{{k}/{r}}
\qquad
\hbox{and}
\qquad
\sigma_{1}(\matU\transp\matS)
\le 1 + \sqrt{{\ell}/{r} }.
$$
Clearly, $\matS$ may be viewed as a matrix that samples and rescales $r$ \textit{rows} of $\matU$ and $\matV$ (columns of $\matU\transp$ and $\matV\transp$),
namely the rows that correspond to non-zero weights $s_i$.
\begin{lemma}[Dual Set Spectral-Frobenius Sparsification.]
\label{lemma:intro2}
Let \math{\cl V=\{\v_1,\ldots,\v_n\}} be a decomposition of the identity, where \math{\v_i\in\R^{k}} ($k < n$) and
$\sum_{i=1}^n\v_i\v_i\transp=\matI_{k}$; let \math{\cl A=\{\a_1,\ldots,\a_n\}} be an arbitrary set
of vectors, where \math{\a_i\in\R^{\ell}}. Then,
given an integer \math{r} such that \math{k < r \le n}, there exists a set of weights \math{s_i\ge 0} ($i=1\ldots n$), {at most \math{r} of which are non-zero}, such that
\eqan{
\lambda_{k}\left(\sum_{i=1}^ns_i\v_i\v_i\transp\right)
\ge
\left(1 - \sqrt{\frac{k}{r}}\right)^2,
\qquad
\trace\left(\sum_{i=1}^n s_i\a_i\a_i\transp\right)
\le
\trace\left(\sum_{i=1}^n \a_i\a_i\transp\right)
=
\sum_{i=1}^n \TNormS{\a_i}.
}
The weights $s_i$ can be computed deterministically in $O\left(rnk^2+n\ell\right)$ time.
\end{lemma}

In matrix notation (here $\matA$ denotes the matrix whose \textit{rows} are the vectors $\a_i$), the above lemma guarantees that
$
\sigma_{k}\left(\matV\transp\matS\right)
\ge
1 - \sqrt{{k}/{r}}$ and
$\FNormS{\matA\transp\matS}
\le \FNormS{\matA}$. Observe that the second condition is actually
just a general
statement about preserving the sum of positive numbers, though
for our specific application context, as stated in the lemma,
this sum happens to be the Frobenius norm of \math{\matA}.


\section{Proofs of our Main Results}
\label{sec:PROOFS}

In this section, we leverage the main tools described in Section~\ref{sec:main1} in order to prove the results of Section~\ref{sec:main0} (Theorems \ref{theorem:intro1} through \ref{thmFast3}). We start with a proof of Theorem~\ref{theorem:intro1},
using Lemmas~\ref{lem:generic} and~\ref{lemma:intro1}.

\subsection{Proof of Theorem~\ref{theorem:intro1}}
Apply the algorithm of Lemma~\ref{lemma:intro1}
on the following two sets of vectors:
the $n$ rows of the matrix $\matV_k$ and the $n$ rows of the matrix
$\matV_{\rho-k}$. The output of the algorithm is a sampling and
rescaling matrix $\matS \in \mathbb{R}^{n \times r}$
(see discussion after Lemma~\ref{lemma:intro1} in
Section~\ref{sec:mainsparse}). Let \math{\matC=\matA\mat S} and note that
$\matC$ consists of a subset of $r$ rescaled
 columns of $\matA$. Lemma \ref{lemma:intro1} guarantees that $\sigma_{k}(\matV_k\transp\matS)\geq 1-\sqrt{k/r}>0$ (assuming $r>k$), and so $\rank(\matV_k\transp\matS)=k$. Also,
\math{\sigma_1(\matV_{\rho-k}\transp\matS)=
\mynorm{\matV_{\rho-k}\transp\matS}_2\leq 1+\sqrt{(\rho-k)/r}
}.
Applying Eqn.~(\ref{eqn3}) of Lemma~\ref{lem:generic}, we obtain,
\eqan{\TNormS{\matA - \Pi_{\matC,k}^2(\matA)}
&\leq&
\TNormS{\matA - \matA_k}+
\TNormS{(\matA-\matA_k) \matS (\matV_k\transp \matS)^+}\\
&\le&
\TNormS{\matA - \matA_k}+
\TNormS{(\matA -\matA_k)\matS}\TNormS{(\matV_k\transp\matS)^+}\\
&=&
\TNormS{\matA - \matA_k}+
\TNormS{\matU_{\rho-k}\matSig_{\rho-k}\matV_{\rho-k}\transp
\matS}\TNormS{(V_k\transp \matS)^+}\\
&\le&
\TNormS{\matA - \matA_k}+
\TNormS{\matSig_{\rho-k}}\TNormS{\matV_{\rho-k}\transp
\matS}\TNormS{(\matV_k\transp \matS)^+}\\
&\le&
\TNormS{\matA - \matA_k}\left(1+\frac{(1+\sqrt{(\rho-k)/r})^2}
{(1-\sqrt{k/r})^2}\right),
}
where the last inequality follows because
\math{\TNorm{\matSig_{\rho-k}}=\TNorm{\matA - \matA_k}} and
\math{\TNorm{(\matV_k\transp \matS)^+}=
1/\sigma_{k}(\matV_k\transp \matS)\le 1/(1-\sqrt{k/r})}.
Theorem \ref{theorem:intro1} now follows by taking square roots of
both sides and using \math{\sqrt{1+x^2}\le 1+x}.
The running time is equal to the time needed to compute
$\matV_k$ and $\matV_{\rho-k}$ plus the running time of the algorithm in
Lemma~\ref{lemma:intro1}. Finally, we note that rescaling the
columns of $\matC$ does not change the span of its columns and
thus is irrelevant
in the construction of $\Pi_{\matC,k}^2(\matA)$.

Our next theorem describes a deterministic algorithm for spectral norm reconstruction that only needs to compute \math{\matV_k} and will serve as a prequel to the proof of Theorem~\ref{thmFast1}.
\begin{theorem}\label{theorem:spectralIn}
Given \math{\matA\in\R^{m\times n}} of rank $\rho$ and a target rank $k < \rho$, there exists a deterministic polynomial-time algorithm to select \math{r > k} columns of \math{\matA} and form a matrix
$\matC\in\R^{m \times r}$ such that
$$\TNorm{\matA - \Pi_{\matC,k}^2(\matA)} \leq
\frac{1 + \sqrt{ n/r }}{1 - \sqrt{ k/r }} \cdot \TNorm{\matA-\matA_k}.$$
The matrix $\matC$ can be computed in \math{T_{\matV_k} + O\left(nrk^2\right)} time, where $T_{\matV_k}$ is the time needed to compute the top $k$ right singular vectors of $\matA$.
\end{theorem}
\begin{proof}
First, apply the algorithm of Lemma~\ref{lemma:intro1}
on the following two sets of vectors: the $n$ rows of the matrix
$\matV_k$ and the $n$ rows of the matrix $\matI_n$. The output of the algorithm is a sampling and rescaling matrix $\matS \in {\R}^{n \times r}$ (see discussion after Lemma~\ref{lemma:intro1} in Section~\ref{sec:mainsparse}). Let \math{\matC=\matA \matS} and note that $\matC$ consists of a subset of $r$ rescaled
 columns of $\matA$. Lemma \ref{lemma:intro1} guarantees that \math{\norm{\matI_n\matS}_2\le 1+\sqrt{n/r}}. Applying Eqn.~(\ref{eqn4}) of Lemma~\ref{lem:generic}, we get
\eqan{ \TNormS{\matA - \Pi^{\xi}_{\matC,k}(\matA)} &\leq& \TNormS{\matA-\matA_k} \cdot \TNormS{ \matS (\matV_k\transp \matS)^+} \\
&=&    \TNormS{\matA-\matA_k} \cdot \TNormS{ \matI_n \matS (\matV_k\transp \matS)^+} \\
&\le&  \TNormS{\matA-\matA_k} \cdot \TNormS{ \matI_n \matS} \cdot \TNormS{ (\matV_k\transp \matS)^+} \\
&\le&  \TNormS{\matA-\matA_k} \cdot \left(1+\sqrt{n/r} \right)^2  \cdot \left( 1 - \sqrt{ k/r } \right)^{-2}.
}
Again, as in Theorem~\ref{theorem:intro1}, the rescaling of the columns of $\matC$ is irrelevant to the construction of $\Pi_{\matC,k}^2(\matA)$. To analyze the running time of the proposed algorithm, we need to look more closely at Lemma~\ref{lemma:intro1} and the related algorithm. The proof of this Lemma in Section~\ref{sec:proofpure1} argues that the algorithm of Lemma~\ref{lemma:intro1} can be implemented in $O(nrk^2)$ time. The total running time is the time needed to compute $\matV_k$ plus $O(nrk^2)$.
\end{proof}

\subsection{Proof of Theorem~\ref{thmFast1}} In order to prove Theorem \ref{thmFast1} we will follow the proof of Theorem~\ref{theorem:spectralIn} using Lemma~\ref{tropp1} (a fast matrix factorization) instead of Lemma~\ref{lem:generic} (the exact SVD of $\matA$). More specifically, instead of using the top \math{k} right singular vectors of \math{\matA} (the matrix \math{\matV_k}), we use the matrix \math{\matZ \in \mathbb{R}^{n \times k}} of Lemma~\ref{tropp1}. We now apply the algorithm of Lemma~\ref{lemma:intro1}
on the following two sets of vectors: the $n$ rows of the matrix $\matZ$ and the $n$ rows of the matrix $\matI_n$. The output of the algorithm is a sampling and rescaling matrix $\matS \in {\R}^{n \times r}$ (see discussion after Lemma~\ref{lemma:intro1} in ection~\ref{sec:mainsparse}). Let \math{\matC=\matA \matS} and note that $\matC$ consists of a subset of $r$ rescaled
 columns of $\matA$.
Applying Eqn.~(\ref{eqn2}) of Lemma~\ref{lem:genericNoSVD}, we get
\eqan{ \TNormS{\matA - \Pi^{\xi}_{\matC,k}(\matA)} &\leq& \TNormS{\matE} \cdot \TNormS{ \matS (\matZ\transp \matS)^+} \\
&=& \TNormS{\matE} \cdot \TNormS{ \matI_n \matS (\matZ\transp \matS)^+} \\
&\le&  \TNormS{\matE} \cdot \TNormS{ \matI_n \matS} \cdot \TNormS{ (\matZ\transp \matS)^+} \\
&\le&  \TNormS{\matE} \cdot \left(1+\sqrt{n/r} \right)^2  \cdot \left( 1 - \sqrt{ k/r } \right)^{-2},
}
where \math{\matE} is the residual error from the matrix factorization of Lemma~\ref{tropp1}. Taking square roots and using the bounds guaranteed by Lemma~\ref{lemma:intro1} for \math{\norm{\matI_n\matS}_2} and \math{\mynorm{(\matZ\transp\matS)^+}_2}, we obtain a bound in terms of \math{\norm{\matE}_2},
$$ \TNorm{\matA - \Pi^{\xi}_{\matC,k}(\matA)} \le \TNorm{\matE} \cdot \left(1+\sqrt{n/r} \right) \cdot \left( 1 - \sqrt{ k/r } \right)^{-1}.$$
Finally, since \math{\matE} is a random variable, taking expectations and applying the bound of Lemma~\ref{tropp1} concludes the proof of the theorem. Again, the rescaling of the columns of $\matC$ is irrelevant to the construction of $\Pi_{\matC,k}^2(\matA)$. The running time is the time needed to compute the matrix $\matZ$ from Lemma~\ref{tropp1} plus an additional $O(nrk^2)$ time as in Theorem~\ref{theorem:spectralIn}.

\subsection{Proof of Theorem \ref{theorem:intro2}}  First,
apply the algorithm of Lemma~\ref{lemma:intro2}
on the following two sets of vectors: the $n$ rows of the matrix $\matV_k$ and the $n$ rows of the matrix $\left(\matA-\matA_k\right)\transp$. The output of the algorithm is a sampling and rescaling matrix $\matS \in {\R}^{n \times r}$ (see discussion after Lemma~\ref{lemma:intro1} in Section~\ref{sec:mainsparse}). Let \math{\matC=\matA \matS} and note that $\matC$ consists of a subset of $r$ rescaled
 columns of $\matA$. We follow the proof of Theorem~\ref{theorem:intro1} in the previous section up to the point where we need to bound the term
\math{\mynorm{(\matA-\matA_k)\matS (\matV_k\transp \matS)^+}^2_{\mathrm{F}}}.
By strong submultiplicativity,
\mand{
\FNorm{(\matA-\matA_k) \matS (\matV_k\transp \matS)^+}^2
\le
\FNorm{(\matA-\matA_k) \matS}^2
\TNorm{(\matV_k\transp \matS)^+}^2.
}
To conclude, we apply Lemma~\ref{lemma:intro2} to bound the two terms in the right-hand side of the above inequality. The rescaling of the columns of $\matC$ is irrelevant to the construction of $\Pi_{\matC,k}^F(\matA)$. The running time of the proposed algorithm is equal to the time needed to compute $\matV_k$ plus the time needed to compute $\matA-\matA_k$ (which is equal to \math{O(mnk)} given $\matV_k$) plus the time needed to run the algorithm of Lemma~\ref{lemma:intro2}, which is equal to \math{O\left(nrk^2+nm\right)}.

\subsection{Proof of Theorem \ref{thmFast2}}
We will follow the proof of Theorem~\ref{theorem:intro2}, but, as with the proof of Theorem~\ref{thmFast1}, instead of using the top \math{k} left singular vectors of \math{\matA} (the matrix \math{\matV_k}), we will use the matrix \math{\matZ} of Lemma~\ref{tropp2} that is computed via a fast, approximate matrix factorization. More specifically, let $\matZ$ be the matrix of Lemma~\ref{tropp2} and run the algorithm of Lemma~\ref{lemma:intro2} on the following two sets of vectors: the $n$ rows of the matrix $\matZ$ and the $n$ rows of the matrix
$\matE\transp$.
The output of the algorithm is a sampling and rescaling matrix
$\matS \in {\R}^{n \times r}$ (see discussion after Lemma~\ref{lemma:intro1} in Section~\ref{sec:mainsparse}). Let \math{\matC=\matA \matS} and note that $\matC$ consists of a subset of $r$ rescaled
 columns of $\matA$. The proof of Theorem \ref{thmFast2} is now identical to the proof of Theorem~\ref{theorem:intro2}, except for using Eqn.~(\ref{eqn2}) of Lemma~\ref{lem:genericNoSVD} instead of Eqn.~(\ref{eqn4}) in Lemma~\ref{lem:generic}. Ultimately, we obtain
\eqan{
\FNormS{\matA - \Pi^{\mathrm{F}}_{\matC,k}(\matA)}
&\leq&
\FNormS{\matE}+
\TNormS{\matE \matS (\matZ\transp \matS)^+}\\
&\le&
\FNormS{\matE}+
\FNormS{\matE\matS}\TNormS{(\matZ\transp \matS)^+}\\
&\le&
\left(1+\left(1-\sqrt{k/r}\right)^{-2}\right)\FNormS{\matE}.
}
The last inequality follows from the bounds of Lemma~\ref{lemma:intro2}.
The theorem now follows by taking the expectation of both sides and
using Lemma~\ref{tropp2} to bound
$\Expect{ \FNormS{\matE} }$.
The rescaling of the columns of~$\matC$ is irrelevant to the construction of $\Pi_{\matC,k}^{\mathrm{F}}(\matA)$. The overall running time is derived by replacing the time needed to compute \math{\mat\matV_k} in Theorem~\ref{theorem:intro2} with the time needed to compute the fast approximate
factorization of Lemma~\ref{tropp2}.

\subsection{Proof of Theorem~\ref{thmFast3}}
Finally, we will prove Theorem~\ref{thmFast3} by combining the results of Theorem~\ref{thmFast2} (a constant factor approximation algorithm) with one round of adaptive sampling. We first recall the following lemma, which has appeared in prior work~\cite{DFKVV99,FKV98}.

\begin{lemma}\label{lem:kar}
Given a matrix \math{\matA\in\R^{m\times n}}, a target rank $k$, and an integer $r$, there exists an algorithm to select $r$ columns from $\matA$ to form the matrix \math{\matC\in\R^{m\times r}} such that
\mand{ \Expect{ \FNormS{\matA - \Pi^{\mathrm{F}}_{\matC,k}(\matA)} }
\le
\FNormS{\matA-\matA_k}+\frac{k}{r}\FNormS{\matA}.
}
The matrix $C$ can be computed in $O(mn + r\log r)$ time.
\end{lemma}

Algorithms for the above lemma
choose $r$ columns of $\matA$ in $r$ independent identically distributed (i.i.d.) trials, where in each trial a column of $\matA$ is sampled with probability proportional to its norm-squared (importance sampling). We now state Theorem 2.1 of~\cite{DRVW06}, which builds upon Lemma~\ref{lem:kar}
to provide an adaptive sampling procedure that improves the accuracy guarantees of Lemma~\ref{lem:kar}.
\begin{lemma}\label{oneround}
Given a matrix $\matA \in \R^{m \times n}$, let $\matC_1 \in \R^{m \times r}$ consist of $r$ columns of $\matA$,
and define the residual
 $\matB = \matA - \matC_1 \matC_1^+ \matA
\in \R^{m \times n}$. For $i=1,\ldots,n$, let
$$p_i = {\TNormS{\b_{i}}}/{\FNormS{\matB}},$$
where $\b_i$ is the $i$-th column of the matrix $\matB$. Sample a further
$s$ columns from $\matA$ in \math{s} i.i.d. trials, where in each trial the $i$-th column is chosen with probability $p_i$.
Let $\matC_2 \in \R^{m \times s}$ contain the $s$ sampled columns and let $\matC = [\matC_1\ \ \matC_2] \in \R^{m \times (r+s)}$ contain the columns of both $\matC_1$ and $\matC_2$, all of which are columns of
\math{\matA}.  Then, for any integer $k > 0$,
$$\Expect{ \mynorm{ \matA - \Pi_{\matC,k}^{\mathrm{F}}(\matA) }_F^2 } \le \FNormS{ \matA - \matA_k } + \frac{k}{s} \norm{\matB}_F^2.$$
\end{lemma}

Note that Lemma~\ref{oneround} is an extension of Lemma~\ref{lem:kar}, which can be derived by setting \math{\matC_1} to be empty in
Lemma~\ref{oneround}.   We are now ready to prove  Theorem~\ref{thmFast3}.
First, fix $d>1$ and define 
$$c_0=\left(1+\epsilon_0\right) \left(1+{1}/{\left(1-\sqrt{k/\hat r} \right)^2} \right),$$
where 
$$\hat r=\ceil{dk}.$$ 
(We will choose \math{d} and \math{\epsilon_0} later.)
Now run the algorithm of Theorem~\ref{thmFast2} to sample
\math{\hat r=\ceil{dk}} columns of $\matA$ and form the matrix \math{\matC_1}. Then, run the adaptive sampling algorithm of Lemma~\ref{oneround} with \math{\matB=\matA-\matC_1\matC_1^+\matA} and
sample a further  $$s = \ceil{c_0k/
\epsilon}$$ columns of $\matA$ to form the matrix $\matC_2$.
Let 
$$\matC = [\matC_1\ \ \matC_2]\in\R^{n\times\left(\hat{r}+s\right)}$$
contain all the
sampled columns. We will analyze the expectation
\mand{\Exp{\mynorm{\matA-\Pi^{\mathrm{F}}_{\matC,k}(\matA)}_{\mathrm{F}}^2}.
}
Using the bound of Lemma~\ref{oneround}, we first compute the expectation
with respect to \math{\matC_2} conditioned on $\matC_1$:
$$
\Exp_{\matC_2}\left[\left.\FNormS{\matA - \Pi^{\mathrm{F}}_{\matC,k}(\matA)}
\right|\matC_1\right]
\le \FNormS{\matA-\matA_k}+\frac{k}{s} \FNormS{\matB}.
$$
We now compute the expectation with respect to \math{\matC_1} (only
\math{\matB} depends on \math{\matC_1}):
\mld{ \Exp_{\matC_1}\left[\Exp_{\matC_2}
\left[\left.\FNormS{\matA - \Pi^{\mathrm{F}}_{\matC,k}(\matA)}
\right|\matC_1\right]
\right] \le
\FNormS{\matA-\matA_k}+\frac{k}{s} \Exp_{\matC_1}\left[\FNormS{\matA-\matC_1\matC_1^+\matA}\right].
\label{eq:iterated}
}
By the law of iterated expectation, the left hand side is equal to \math{\Expect{\FNormS{\matA-\Pi^{\mathrm{F}}_{\matC,k}(\matA)}}}.
We now use the accuracy guarantee of Theorem~\ref{thmFast2} and
our definition of $c_0$ 
to bound
$$\Exp_{\matC_1}\left[\FNormS{\matA-\matC_1\matC_1^+\matA}\right]
\le \Exp_{\matC_1}\left[\FNormS{\matA-\Pi^{\mathrm{F}}_{C_1,k}(\matA)}\right] \le c_0
\FNormS{\matA-\matA_k}.$$
Using the bound in \r{eq:iterated}, we obtain
\mand{
\Expect{\FNormS{\matA-\Pi^{\mathrm{F}}_{\matC,k}(\matA)}}
\le
\FNormS{\matA-\matA_k}\left(1+{c_0k}/{s}\right).
}
Finally, recall that for our choice of \math{s},  \math{s\ge c_0k/\epsilon},
and so we obtain
the relative error bound. The number of columns needed is
$$r=\hat r+s=dk+c_0k/\epsilon.$$ 
Set
$$d=(1+\alpha)^2,$$ where
$$\alpha=\sqrt[3]{(1+\epsilon_0)/\epsilon}.$$ After some algebra,
this
yields
$$
r=k\left( \alpha^3+\left(1+\alpha\right)^3\right)=\frac{2k}{\epsilon}\left(1+O\left(\epsilon_0+\epsilon^{1/3} \right) \right)
$$
sampled columns.
The time needed to compute the matrix $\matC$ is the sum of three terms:
the running time of Theorem~\ref{thmFast2} (which is
\math{O(mnk\epsilon_0^{-1}+n\hat rk^2)}),
plus the time needed to compute \math{\matA-\matC_1\matC_1^+\matA}
(which is \math{O(mn\hat r)}), plus the time needed to run the algorithm
of Lemma~\ref{oneround} (which is \math{O(mn+s\log s)}).
Assume \math{r<n} (otherwise the problem is trivial),
set 
$$
\epsilon_0=\epsilon^{2/3},
$$
and use~\math{d=O(\epsilon^{-2/3})} to get the final
asymptotic run time.

\subsubsection{Improving the running time}
We conclude by noting that the number of columns required for relative error approximation is approximately~\math{\frac{2k}{\epsilon}}, a
two-factor from optimal, since \math{\frac{k}{\epsilon}} columns are necessary (see \cite{DV06} and
Section~\ref{sec:lower}). We get an improved running time equal to $$O\left(mnk+nk^3+n\log\epsilon^{-1}\right)$$ using
just a constant factor more columns by setting
\math{d} and \math{\epsilon_0} in the proof to constants (for example, setting
\math{d=100,\ \epsilon_0=\frac{62}{181}\approx\frac13} results in sampling
\math{\frac{3k}{\epsilon}(1+o(1))} columns).


\section{Proof of Lemma \ref{tropp1}: Approximate SVD in the Spectral Norm}\label{sec:proofS}
Consider the following algorithm, described in Corollary 10.10 of~\cite{HMT}. The algorithm takes as inputs a matrix $\matA \in \R^{m \times n}$ of rank $\rho$, an integer $2 \leq k <\rho$, an integer $q \geq 1$, and an integer $p \geq 2$. Set $r = k+p$ and construct the matrix
$\matY \in \R^{m \times r}$ as follows:
\begin{enumerate}
\item Generate an $n \times r$ standard Gaussian matrix $\matR$ whose entries are i.i.d. $\mathcal{N}(0,1)$ variables.
\item Return $\matY = (\matA \matA\transp)^q \matA \matR \in \R^{m \times r}$.
\end{enumerate}
The running time of the above algorithm is $ O(mnrq )$.
Corollary 10.10 of~\cite{HMT} presents the following bound:
$$\Expect{\TNorm{\matA - \matY \matY^+ \matA }} \leq
\left(1 + \sqrt{\frac{k}{p-1}} + \frac{e \sqrt{k+p}}{p} \sqrt{ \min\{ m,n\} -k }
\right)^{\frac{1}{2q+1}}\TNorm{\matA - \matA_k},$$
where $e = 2.718\ldots$. To the best of our understanding, the above result is not immediately applicable to the construction of a factorization of the form $\matA =
\matB \matZ\transp + \matE$, because $\matY$ contains $r > k$ columns. Lemma \ref{troppextension0} below, which  strengthens Corollary 10.10 in \cite{HMT}, argues that the matrix $\Pi_{\matY,k}^2(\matA)$ contains the desired factorization $\matB \matZ\transp$. Recall that while we cannot compute $\Pi_{\matY,k}^2(\matA)$ efficiently, we can compute a constant-factor approximation, which is sufficient for our purposes. The proof of Lemma~\ref{troppextension0} is very similar to the proof of Corollary 10.10 of~\cite{HMT}, with the only difference being our starting point: instead of using Theorem 9.1 of~\cite{HMT} we use Lemma~\ref{lem:generic} of our work. To prove Lemma~\ref{troppextension0}, we will
need several results for standard Gaussian matrices,
projection matrices, and H\"{o}lder's inequality.
The following seven lemmas are all borrowed from \cite{HMT}.

\begin{lemma} [Proposition 10.1 in \cite{HMT}] \label{prop1b}
Fix matrices $\matX$, $\matY$, and draw a standard Gaussian matrix $\matR$ of appropriate dimensions. Then,
\mand{
\qquad \Expect{ \TNorm {\matX \matR \matY} }  \leq \TNorm{\matX} \FNorm{\matY} + \FNorm{\matX} \TNorm{\matY}.
}
\end{lemma}
\begin{lemma} [Proposition 10.2 in \cite{HMT}] \label{prop2b}
For $k, p \geq 2$, draw a standard Gaussian matrix $\matR \in \R^{k \times (k+p)}$. Then,
\mand{
\qquad \Expect{ \TNorm{\matR^+} }  \leq \frac{e \sqrt{k+p}}{p},
}
where $e=2.718\ldots$.
\end{lemma}

\begin{lemma} [Proposition 10.1 in \cite{HMT}] \label{prop1a}
Fix matrices $\matX$, $\matY$, and a standard Gaussian matrix $\matR$ of appropriate dimensions. Then,
\mand{
\Expect{ \FNormS{\matX \matR \matY} } = \FNormS{\matX} \FNormS{Y}.
}
\end{lemma}

\begin{lemma} [Proposition 10.2 in \cite{HMT}] \label{prop2a}
For $k, p \geq 2$, draw a standard Gaussian matrix $\matR \in \R^{k \times (k+p)}$. Then,
\mand{
 \Expect{ \FNormS{\matR^+} }  =\frac{k}{p-1}.
}
\end{lemma}
\begin{lemma} [proved in \cite{HMT}] \label{prop3}
For integers $k,p \geq 1$, and a standard Gaussian $\matR \in \R^{k \times (k+p)}$ the rank of $\matR$ is equal to $k$ with probability one.
\end{lemma}
\begin{lemma} [Proposition 8.6 in \cite{HMT}] \label{projection}
Let $\matP$ be a projection matrix. For any matrix $\matX$ of appropriate dimensions and an integer $q\ge 0$,
$$ \TNorm{\matP \matX} \le \left( \TNorm{ \matP (\matX \matX\transp)^q \matX  }  \right)^{\frac{1}{2q+1}} $$
\end{lemma}
\begin{lemma} [H\"{o}lder's inequality] \label{prop0}
Let $x$ be a positive random variable. Then, for any $h \ge 1$,
$$\Expect{ x } \leq \left( \Expect{ x^h } \right)^{\frac{1}{h}}.$$
\end{lemma}
\noindent The following lemma provides an alternative definition for $\Pi_{\matC,k}^{\xi}(\matA)$ which will be useful in subsequent proofs.
Recall from Section~\ref{sec:bestrankk} that we can write
$\Pi_{\matC,k}^\xi(\matA)= \matC\matX^\xi$, where
$$
\matX^\xi =
\argmin_{\Psi \in \mathbb{R}^{r \times n}:\rank(\Psi)\leq k}\XNormS{\matA-\matC\Psi}.
$$
The next lemma basically says that $\Pi_{\matC,k}^\xi(\matA)$ is the projection
of \math{\matA} onto the rank-\math{k} subspace spanned by
\math{\matC\matX^\xi}, and that no other subspace in the column space of
\math{\matC}   is better.
\begin{lemma}\label{lem:PPD1}
For $\matA \in {\R}^{m \times n}$ and $\matC \in {\R}^{m \times r}$,
integer $r> k$, let
$\Pi_{\matC,k}^\xi(\matA)= \matC\matX^\xi$, and
$\matY\in \mathbb{R}^{r \times n}$ be any matrix of rank at most $k$.
Then,
$$\XNormS{\matA-\matC\matX^\xi} = \XNormS{\matA-(\matC\matX^\xi)
(\matC\matX^\xi)^+\matA} \leq \XNormS{\matA-(\matC\matY)(\matC\matY)^+\matA},$$
where $\matY \in \mathbb{R}^{r \times n}$ is any matrix of rank at most $k$.
\end{lemma}

\begin{proof}
The second inequality will follow from the optimality of
\math{\matX^\xi} because
\math{\matY(\matC\matY)^+\matA} has rank at most \math{k}. So we only need to
prove the first equality. Again, by the optimality of
\math{\matX^\xi} and  because
\math{\matX^\xi(\matC\matX^\xi)^+\matA} has rank at most \math{k},
$$\XNormS{\matA-\matC\matX^\xi} \le \XNormS{\matA-(\matC\matX^\xi)
(\matC\matX^\xi)^+\matA}.$$ To get the reverse inequality, we will use
matrix-Pythagoras as follows:
\begin{eqnarray*}
\XNormS{\matA - \matC\matX^\xi}
&=&
\XNormS{\left(\matI_m-(\matC\matX^\xi)(\matC\matX^\xi)^+\right)\matA-
\matC\matX^\xi(\matI_n-(\matC\matX^\xi)^+\matA)}\\
&\ge&
\XNormS{\left(\matI_m-(\matC\matX^\xi)(\matC\matX^\xi)^+\right)\matA}.
\end{eqnarray*}
\end{proof}

\begin{lemma} [Extension of Corollary 10.10 of \cite{HMT}] \label{troppextension0}
Let $\matA$ be a matrix in $\R^{m \times n}$ of rank $\rho$,  let $k$ be an integer satisfying $2 \leq k < \rho$, and let $r = k+p$ for some integer $p \geq 2$.
Let $\matR \in \R^{n \times r}$ be a standard Gaussian matrix (i.e., a matrix whose entries are drawn in i.i.d. trials from $\mathcal{N}(0,1)$).
Define $\matB = (\matA \matA\transp)^q \matA$ and compute $\matY = \matB \matR$. Then, for any $q \geq 0$,
$$\Expect{\TNorm{\matA - \Pi_{\matY,k}^2(\matA) } } \leq
\left(1 + \sqrt{\frac{k}{p-1}} + \frac{e \sqrt{k+p}}{p} \sqrt{ \min\{ m,n\} -k } \right)^{\frac{1}{2q+1}}
\TNorm{\matA - \matA_k}$$
\end{lemma}
\begin{proof}
Let $\Pi_{\matY,k}^2(\matA) =\mat Y\matX_1$ and $\Pi_{\matY,k}^2(\matB)
= \matY\matX_2$, where
\math{\matX_1} is optimal for \math{\matA} and \math{\matX_2} for
\math{\matB}.
From Lemma~\ref{lem:PPD1},
$$\TNorm{\matA-\Pi_{\matY,k}^2(\matA)} =
\TNorm{(\matI_m-(\matY\matX_1)(\matY\matX_1)^+)\matA} \leq
\TNorm{(\matI_m-(\matY\matX_2)(\matY\matX_2)^+)\matA}.$$
From Lemma~\ref{projection} and using the fact that $\matI_m-(\matY\matX_2)(\matY\matX_2)^+$ is a projection,
\begin{eqnarray*}
\TNorm{(\matI_m-(\matY\matX_2)(\matY\matX_2)^+)\matA}&\leq&
\TNorm{\left(\matI_m-\left(\matY\matX_2\right)\left(\matY\matX_2\right)^+\right)\left(\matA\matA\transp\right)^q\matA}^{\frac{1}{2q+1}}\\
&=&
\TNorm{\matB-\left(\matY\matX_2\right)\left(\matY\matX_2\right)^+\matB}^{\frac{1}{2q+1}}\\
&=&
\TNorm{\matB-\Pi_{\matY,k}^2(\matB)}^{\frac{1}{2q+1}},
\end{eqnarray*}
where the last step follows from Lemma~\ref{lem:PPD1}.
We conclude that
\begin{equation*}
\TNorm{\matA-\Pi_{\matY,k}^2(\matA)} \leq \TNorm{\matB-\Pi_{\matY,k}^2(\matB)}^{\frac{1}{2q+1}}.
\end{equation*}
The matrix \math{\matY} is generated using a random \math{\matR}, so
taking expectations and applying H\"{o}lder's inequality, we get
\begin{equation}\label{eqn:PPD22}
\Expect{\TNorm{\matA-\Pi_{\matY,k}^2(\matA)}} \leq \left(\Expect{\TNorm{\matB-\Pi_{\matY,k}^2(\matB)}}\right)^{\frac{1}{2q+1}}.
\end{equation}
We now focus on bounding the term on
the right-hand side of the above equation.
Let the SVD of $\matB$ be $\matB = \matU_{\matB} \matSig_{\matB} \matV_{\matB}\transp$, with
the top rank $k$ factors from the SVD of $\matB$ being
 $\matU_{\matB,k}$, $\matSig_{\matB,k}$, and $\matV_{\matB,k}$ and
 the  corresponding trailing factors being
$\matU_{\matB,\tau}$, $\matSig_{\matB,\tau}$ and $\matV_{\matB,\tau}$.
Let $\rho_\matB$ be the rank of $\matB$.
Let
\mand{
       \matOmega_1 = \matV_{\matB,k}\transp \matR \in \R^{k \times r} \qquad \mbox{and} \qquad
\qquad \matOmega_2 = \matV_{\matB,\tau}\transp \matR \in \R^{(\rho_{\matB}-k) \times r}.
}
The Gaussian distribution is rotationally invariant, so $\matOmega_1$, $\matOmega_2$ are also standard Gaussian matrices which are
stochastically independent because $\matV_{\matB}\transp$ can be extended to a
full rotation. Thus, $\matV_{\matB,k}\transp \matR$ and $\matV_{\matB,\tau}\transp\matR$ also have entries that are i.i.d. $\mathcal{N}(0,1)$ variables. We now apply Lemma~\ref{lem:generic} to reconstruct \math{\matB},  with $\xi = 2$ and $\matS = \matR$.
The rank requirement in
 Lemma~\ref{lem:generic} is satisfied because,
from Lemma~\ref{prop3}, the rank of $\matOmega_1$ is equal to $k$ (as it is a
standard normal matrix), and thus the matrix $\matR$
satisfies the rank assumptions of Lemma \ref{lem:generic}. We get that,
$$\TNormS{ \matB - \Pi^{2}_{\matY,k}(\matB) }
\leq  \TNormS{\matB-\matB_k} + \TNormS{(\matB-\matB_k)\matR(\matV_{\matB,k}\transp\matR)^+ }
\le
\TNormS{\matSig_{\matB,\tau}} +
\TNormS{\matSig_{\matB,\tau} \matOmega_2 \matOmega_1^+ }.
$$
Using \math{\sqrt{x^2+y^2}\le x+y}, we conclude that
$$
\TNorm{ \matB - \Pi^{2}_{\matY,k}(\matB) }\le
\TNorm{\matSig_{\matB,\tau}} +
\TNorm{\matSig_{\matB,\tau} \matOmega_2 \matOmega_1^+ }.
$$
We now need to take the expectation with respect to
\math{\matOmega_1,\matOmega_2}. We first take the expectation with respect to
\math{\matOmega_2}, conditioning on \math{\matOmega_1}. We then take the expectation
with respect to \math{\matOmega_1}. Since only the second term is stochastic, using
Lemma~\ref{prop1b}, we have:
\begin{eqnarray}
{\bf E}_{\matOmega_2}\left[ \TNorm{ \matSig_{\matB,\tau} \matOmega_2 \matOmega_1^+ }|\matOmega_1\right]
\nonumber &\leq&   \TNorm{ \matSig_{\matB,\tau}} \FNorm{\matOmega_1^+ }
+ \FNorm{ \matSig_{\matB,\tau}} \TNorm{\matOmega_1^+ }.
\end{eqnarray}
We now take the expectation with respect to \math{\matOmega_1}.
To bound \math{\Expect{\TNorm{\matOmega_1^+ }}}, we use Lemma~\ref{prop2b}.
To bound \math{\Expect{\FNorm{\matOmega_1^+ }}}, we first use H\"{o}lder's
inequality to bound $$\Expect{\FNorm{\matOmega_1^+ }}\le
\Expect{\FNormS{\matOmega_1^+ }}^{1/2},$$ and then we use
Lemma~\ref{prop2a}.
Since  
$$\FNorm{ \matSig_{\matB,\tau}}\le\sqrt{\min(m,n)-k}
\TNorm{ \matSig_{\matB,\tau}},
$$
collecting our results together, we obtain:
\eqan{
\Expect{\TNorm{\matB-\Pi_{\matY,k}^2(\matB)}}
\le
\left(1+\sqrt{\frac{k}{p-1}}+\frac{e \sqrt{k+p}}{p}\sqrt{\min(m,n)-k}
\right)
\TNorm{\matB-\matB_k}
.}
To conclude, combine with Eqn.~(\ref{eqn:PPD22}) and note that
$$\TNorm{\matB-\matB_k}=\TNorm{\matA-\matA_k}^{2q+1}.$$
\end{proof}

We now have all the necessary ingredients to prove Lemma \ref{tropp1}. Let $\matY$ be the matrix of Lemma~\ref{troppextension0}. Set
$p = k$ and $$q = \ceil{ \frac{ \log\left( 1+
\sqrt{\frac{k}{k-1}} + \frac{e \sqrt{2k}}{k} \sqrt{ \min\{ m,n\} -k } \right)
}{2 \log\left(1+\epsilon/\sqrt{2}\right)  - 1/2 }},$$
so that
$$ \left(1 + \sqrt{\frac{k}{p-1}} + \frac{e \sqrt{k+p}}{p} \sqrt{ \min\{ m,n\} -k } \right)^{\frac{1}{2q+1}} \leq 1 + \frac{\epsilon}{\sqrt{2}}.$$
Then,
\mld{\Expect{\TNorm{\matA - \Pi_{\matY,k}^2(\matA) }} \leq \left(1 + \frac{\epsilon}{\sqrt{2}}\right) \TNorm{\matA - \matA_k}.
\label{eq:expS}
}
Given \math{\matY}, let \math{\matQ} be an orthonormal basis for its column
space. Then,
using the algorithm of Section~\ref{sec:bestrankk}
and applying Lemma~\ref{lem:bestF} we can construct
the matrix $\matQ(\matQ\transp\matA)_k$ such that
\mand{
\TNorm{\matA -  \matQ(\matQ\transp\matA)_k} \leq \sqrt{2} \TNorm{\matA - \Pi_{\matY,k}^2(\matA) }.
}
Clearly, $\matQ(\matQ\transp\matA)_k$ is a rank $k$ matrix;
let $\matZ \in {\R}^{n \times k}$ denote the
matrix containing the right singular vectors of
$\matQ(\matQ\transp\matA)_k$, so
\math{\matQ(\matQ\transp\matA)_k=\matX\matZ\transp}.
Note that $\matZ$ is equal to the right singular vectors of the
matrix $(\matQ\transp\matA)_k$
(because $\matQ$ has orthonormal columns), and so $\matZ$
has already been computed at the second step of the
algorithm of Section~\ref{sec:bestrankk}. Since
\math{\matE=\matA - \matA \matZ
\matZ\transp} and
\math{\TNorm{\matA-\matA\matZ\matZ\transp}\le
\TNorm{\matA-\matX\matZ\transp}} for any \math{\matX}, we have
$$\TNorm{\matE} \leq \TNorm{\matA - \matQ(\matQ\transp\matA)_k}
\le\sqrt{2} \TNorm{\matA - \Pi_{\matY,k}^2(\matA) }.$$
Note that, by construction, $\matE\matZ=\bm0_{m\times k}$.
The running time follows by adding the running time of
the algorithm at the beginning of this section and the
 running time of the algorithm of Lemma~\ref{sec:bestrankk}.

\section{Proof of Lemma \ref{tropp2}: Approximate SVD in the Frobenius Norm}\label{sec:proofF}
Consider the following algorithm, described in Theorem 10.5 of~\cite{HMT}. The algorithm takes as inputs a matrix $\matA \in \R^{m \times n}$ of rank $\rho$, an integer $2 \leq k < \rho$, and an integer $p \geq 2$. Set $r = k+p$ and construct the matrix
$\matY \in \R^{m \times r}$ as follows:
\begin{enumerate}
\item Generate an $n \times r$ standard Gaussian matrix $\matR$ whose entries are i.i.d. $\mathcal{N}(0,1)$ variables.
\item Return $\matY = \matA \matR \in \R^{m \times r}$.
\end{enumerate}
The running time of the above algorithm is $ O(mnr)$.
Theorem 10.5 in \cite{HMT} presents the following bound:
$$\Expect{\FNorm{\matA - \matY\matY^+\matA }} \leq \left( 1 + \frac{k}{p-1} \right)^{\frac{1}{2}} \FNorm{\matA - \matA_k}.$$
To the best of our understanding, the above result is not immediately applicable to the construction of a factorization of the form $\matA =
\matB \matZ\transp + \matE$ (as in Lemma~\ref{tropp2}) because $\matY$ contains $r > k$ columns.
\begin{lemma} [Extension of Theorem 10.5 of~\cite{HMT}] \label{troppextension}
Let $\matA$ be a matrix in $\R^{m \times n}$ of rank $\rho$, let $k$ be an integer satisfying $2 \leq k < \rho$, and let $r = k+p$ for some integer $p \geq 2$.
Let $\matR \in \R^{n \times r}$ be a standard Gaussian matrix (i.e., a matrix whose entries are drawn in i.i.d. trials from $\mathcal{N}(0,1)$)
and compute $\matY = \matA \matR$. Then,
$$\Expect{\FNormS{\matA - \Pi_{\matY,k}^{\mathrm{F}}(\matA) }} \leq
\left( 1 + \frac{k}{p-1} \right) \FNormS{\matA - \matA_k}.$$
\end{lemma}
\begin{proof}
We construct the matrix $\matY$ as described in the beginning of this section. Let the rank of $\matA$ be $\rho$ and let $\matA = \matU \matSig \matV\transp$ be the SVD of $\matA$. Define
\mand{\matOmega_1 = \matV_{k}\transp \matR \in \R^{k \times r} \qquad \mbox{and} \qquad
\qquad \matOmega_2 = \matV_{\rho-k}\transp \matR \in \R^{(\rho-k) \times r}.
}
The Gaussian distribution is rotationally invariant, so $\matOmega_1$, $\matOmega_2$ are also standard Gaussian matrices which are
stochastically independent because $\matV\transp$ can be
extended to a full rotation.
Thus, $\matV_{k}\transp \matR$ and $\matV_{\rho-k}\transp\matR$
also have entries that are i.i.d. $\mathcal{N}(0,1)$ variables.
We now apply Lemma~\ref{lem:generic} to
reconstructing \math{\matA}, with $\xi = F$ and $\matS = \matR$.
Recall that from Lemma~\ref{prop3}, the rank of $\matOmega_1$ is equal to $k$, and thus the matrix $\matR$ satisfies the rank assumptions
of Lemma~\ref{lem:generic}. We have that,
$$\FNormS{ \matA - \Pi_{\matY,k}^{\mathrm{F}}(\matA) } \le
\FNormS{\matA - \matA_k} + \FNormS{ \matSig_{\rho-k} \matOmega_2 \matOmega_1^+} ,$$
where
$\matA - \matA_k = \matU_{\rho-k}  \matSig_{\rho-k} \matV_{\rho-k}\transp$. To conclude, we take the expectation on both sides, and since only
the second term on the right hand side is stochastic, we bound as follows:
\eqan{
\Expect{\FNormS{\matSig_{\rho-k} \matOmega_2 \matOmega_1^+}}
&\buildrel{(a)}\over{=}&
{\bf E}_{\matOmega_1}\left[{\bf E}_{\matOmega_2}\left[
\FNormS{ \matSig_{\rho-k} \matOmega_2 \matOmega_1^+  }|
\matOmega_1   \right]\right]\\
&\buildrel{(b)}\over{=}&
{\bf E}_{\matOmega_1}\left[\FNormS{ \matSig_{\rho-k}} \FNormS{ \matOmega_1^+} \right] \\
&\buildrel{(c)}\over{=}&
\FNormS{ \matSig_{\rho-k}} \Expect{ \FNormS{ \matOmega_1^+} } \\
&\buildrel{(d)}\over{=}&
\frac{k}{p-1} \FNormS{\matSig_{\rho-k}}.
}
\math{(a)} follows from the law of iterated expectation; \math{(b)} follows from Lemma \ref{prop1a}; \math{(c)} follows because $\FNormS{\matSig_{\rho-k}}$ is a constant; \math{(d)} follows from Lemma \ref{prop2a}. We conclude the proof by noting that $\FNorm{\matSig_{\rho-k}} = \FNorm{\matA - \matA_k}$.
\end{proof}

We now have all the necessary ingredients
to conclude the proof of Lemma \ref{tropp2}. Let \math{\matY} be the matrix of
Lemma~\ref{troppextension},
and let \math{\matQ} be an orthonormal basis for its column
space. Then,
using the algorithm of Section~\ref{sec:bestrankk}
and applying Lemma~\ref{lem:bestF} we can construct the matrix
$\matQ(\matQ\transp\matA)_k$ such that
$$\FNormS{\matA -  \matQ(\matQ\transp\matA)_k} =
\FNormS{\matA - \Pi_{\matY,k}^{\mathrm{F}}(\matA) }.$$
Clearly, $ \matQ(\matQ\transp\matA)_k$ is a rank $k$ matrix; let
$\matZ \in {\R}^{n \times k}$ be
the matrix containing the right singular vectors of
$\matQ(\matQ\transp\matA)_k$, so
\math{\matQ(\matQ\transp\matA)_k=\matX\matZ\transp}.
Note that  $\matZ$ is equal to the right singular vectors of the
matrix $(\matQ\transp\matA)_k$ (because $\matQ$ has orthonormal columns), and thus $\matZ$ has already been computed at the second step
of the algorithm of Section~\ref{sec:bestrankk}.  Since
\math{\matE=\matA - \matA \matZ
\matZ\transp} and
\math{\FNorm{\matA-\matA\matZ\matZ\transp}\le
\FNorm{\matA-\matX\matZ\transp}} for any \math{\matX}, we have
$$\FNormS{\matE} \leq \FNormS{\matA - \matQ(\matQ\transp\matA)_k}
=\FNormS{\matA - \Pi_{\matY,k}^2(\matA) }.$$
To conclude, take expectations on both sides,
use Lemma~\ref{troppextension} to bound the term \math{\Expect{\FNormS{\matA - \Pi_{\matY,k}^2(\matA) }}},
and set $p = \ceil{ \frac{k}{\epsilon } + 1}$ to obtain:
$$\Expect{\FNormS{\matE}} \leq \left( 1 + \frac{k}{p-1} \right) \FNormS{\matA - \matA_k}
\leq \left(1+\epsilon\right)\FNormS{\matA - \matA_k},$$
By construction,  $\matE\matZ=\bm0_{m\times k}$.
The running time follows by adding the running time of the algorithm at the beginning of this section and the running time of the algorithm of Lemma~\ref{sec:bestrankk}.

\section{Dual Set Spectral Sparsification: proof of  Lemma~\ref{lemma:intro1}} \label{sec:proofpure1}

In this section, we prove Lemma~\ref{lemma:intro1}, which generalizes Theorem 3.1 in~\cite{BSS09}. Indeed, setting
\math{\cl V=\cl U} reproduces the spectral sparsification result of
 Theorem 3.1 in~\cite{BSS09}.
The basic observation is that the abalysis of~\cite{BSS09} for the upper bound
and lower bound are essentially independent. This means that the
analysis in~\cite{BSS09} \emph{directly} applies when the upper and lower
bounds are analyzed with respect to \emph{different} sets of vectors.
The only point at which the bounds are used together is when
one needs to compute a weight that is in between the two bounds, which
as~\cite{BSS09} show, is always possible with a single set of vectors
and it is also true with different sets of vectors for the same reason.
For completeness we present the details, simplifying a little and
emphasizing the independent treatment of \math{\v_i} and \math{\u_i} in the
proof.

As in~\cite{BSS09}, we will provide a constructive proof of the lemma and
we start by describing the algorithm that computes the weights
$s_i$, $i=1,\ldots,n$.
%
\begin{algorithm}[t]
\textbf{Input:}
     \begin{itemize}
          \item $ \cl V=\{\v_1,\ldots,\v_n\}$, with $\sum_{i=1}^{n}\v_i\v_i\transp=\matI_{k}$ ($k < n$)
          \item $ \cl U=\{\u_1,\ldots,\u_n\}$, with $\sum_{i=1}^{n}\u_i\u_i\transp=\matI_{\ell}$ ($\ell \leq n$)
          \item integer $r$, with $k < r < n$
     \end{itemize}

     \textbf{Output:} A vector of weights $\s=[s_1,\ldots,s_n]$, with $s_i \geq 0$ and at most $r$ non-zero $s_i$'s.

\begin{enumerate}

  \item Initialize $\s_0 = \textbf{0}_{n \times 1}$, $\matA_0 = \textbf{0}_{k \times k}$, $\matB_0 = \textbf{0}_{\ell \times \ell}$.

  \item For $\tau = 0,...,r-1$

        \begin{itemize}

           \item Compute $\scl_\tau$ and $\scu_\tau$ from Eqn.~(\ref{eqn:PD1}).

            \item Find an index $j$ in $\left\{1,\ldots,n\right\}$ such that
			 \eqar{
			 \label{eqn:algo}
			 U(\u_j,\delta_\scu,\matB_\tau,\scu_\tau) \le L(\v_j,\delta_\scl, \matA_\tau,\scl_\tau).
				}		
    \item Let
    \begin{equation}\label{eqn:weight}
    t^{-1} = \frac{1}{2}\left(U(\u_j,\delta_\scu,\matB_\tau,\scu_\tau) + L(\v_j,\delta_\scl, \matA_\tau,\scl_\tau)\right).
    \end{equation}

    \item Update the \math{j}th
component of \math{\s}, \math{\matA_\tau} and \math{\matB_\tau}:
\mld{
\s_{\tau+1}[j] =\s_{\tau}[j] + t,\ \matA_{\tau + 1} = \matA_{\tau} + t \v_j \v_j\transp,\ \hbox{and}\ \matB_{\tau+1} = \matB_{\tau} + t \u_j \u_j\transp.
\label{eq:update}
}

         \end{itemize}

  \item Return $\s = r^{-1}\left(1-\sqrt{k/r}\right)\cdot \s_r$.

\end{enumerate}
\caption{Deterministic Dual Set Spectral Sparsification.}
\label{alg:2set}
\end{algorithm}

\subsection{The Algorithm}
The fundamental idea underlying Algorithm~\ref{alg:2set} is the greedy selection of vectors that satisfy a number of desired properties in each step. These properties will eventually imply the eigenvalue bounds of Lemma~\ref{lemma:intro1}. We start by defining several quantities that will be used in the description of the algorithm and its proof. First, fix two constants:
$$
\delta_\scl = 1;
\qquad
\delta_\scu = \frac{1+\sqrt{\frac{\ell}{r}}}{1-\sqrt{\frac{k}{r}}}.
$$
Given $k$, $\ell$, and $r$ (all inputs of Algorithm~\ref{alg:2set}), and a parameter $\tau = 0,\ldots,r-1$, define two
parameters $\scl_\tau$ and $\scu_\tau$ as follows:
\begin{equation}\label{eqn:PD1}
\scl_\tau = r\left(\frac{\tau}{r}-\sqrt{\frac{k}{r}}\right) = \tau - \sqrt{rk};\\
\scu_\tau = \frac{(\tau-r)\left(1+\sqrt{\frac{\ell}{r}}\right)
               +r\left(1+\sqrt{\frac{\ell}{r}}\right)^2}
               {1-\sqrt{\frac{k}{r}}}=\delta_{\scu}\left(\tau + \sqrt{\ell r}\right).
\end{equation}
We next define the lower and upper
functions \math{\phil(\scl,\matA)} ($\scl\in\R$ and $\matA\in\R^{k \times k}$) and
$\phiu(\scu,\matB)$ ($\scu\in\R$ and $\matB\in\R^{\ell \times \ell}$) as follows:
\begin{equation}\label{eqn:PD2}
\phil(\scl, \matA) =  \sum_{i=1}^k\frac{1}{\lambda_i(\matA)-\scl};
\qquad
\phiu(\scu, \matB) =  \sum_{i=1}^\ell\frac{1}{\scu-\lambda_i(\matB)}.
\end{equation}
Let $L(\v, \delta_\scl, \matA, \scl)$  be a function with four inputs (a vector $\v \in \mathbb{R}^{k\times 1}$, $\delta_\scl \in \mathbb{R}$, a matrix $\matA \in \mathbb{R}^{k \times k}$, and $\scl \in \mathbb{R}$):
\begin{equation}\label{eqn:PD3}
L(\v, \delta_\scl, \matA, \scl)  =  \frac{\v\transp(\matA-(\scl + \delta_\scl)\matI_k)^{-2}\v}
{\phil(\scl + \delta_\scl, \matA)-\phil(\scl,\matA)} -\v\transp(\matA-(\scl + \delta_\scl)\matI_k)^{-1}\v.
\end{equation}
Similarly, let $U(\v, \delta_\scu, \matB, \scu)$  be a function with four inputs (a vector $\u \in \mathbb{R}^{\ell\times 1}$, $\delta_\scu \in \mathbb{R}$, a matrix $\matB \in \mathbb{R}^{\ell \times \ell}$, and $\scu \in \mathbb{R}$):
\begin{equation}\label{eqn:PD4}
U(\u, \delta_\scu, \matB, \scu )  = \frac{\u\transp((\scu + \delta_\scu)\matI_\ell-\matB)^{-2}\u}
{\phiu(\scu, \matB)-\phiu(\scu + \delta_\scu,\matB)}
+\u\transp((\scu + \delta_\scu)\matI_\ell-\matB)^{-1}\u.
\end{equation}
Algorithm \ref{alg:2set} runs in $r$ steps. The initial vector of weights $\s_0$ is initialized to the all-zero vector. At each step $\tau=0,\ldots,r-1$, the
algorithm selects a pair of vectors $(\u_j, \v_j)$ that satisfy Eqn.~(\ref{eqn:algo}), computes the associated weight $t$ from Eqn.~(\ref{eqn:weight}), and updates two matrices and the vector of weights appropriately, as specified
in Eqn. \r{eq:update}.

\subsection{Running time}
\label{sec:runningtime2setS}

The algorithm runs in $r$ iterations. In each iteration, we evaluate the functions $U(\u,\delta_\scu,\matB,\scu)$ and $L(\v,\delta_\scl, \matA,\scl)$ at most $n$ times. Note that all $n$ evaluations for both functions need at most $O(k^3 + nk^2 + \ell^3 + n\ell^2)$ time, because the matrix inversions can be performed once for all $n$ evaluations. Finally, the updating step needs an additional $O(k^2 + \ell^2)$ time. Overall, the complexity of the algorithm is of the order $O(r(k^3 + nk^2 + \ell^3 + n\ell^2 + k^2 + \ell^2) ) = O\left(rn \left(k^2+\ell^2\right)\right)$.

Note that when \math{\cl U} is the standard basis
(\math{\cl U=\{\e_1,\ldots,\e_n\}} and \math{\ell=n}),
the computations can be done much more efficiently:
the
eigenvalues of \math{\matB_\s} need not be computed
explicitly (the expensive step),
since they are available by inspection, being equal to the
weights \math{\s_\tau}.
In the function \math{U(\u,\delta_{\scu},\matB,\scu)},
the functions \math{\phiu} (given the eigenvalues) need only be computed
once per iteration, in \math{O(n)} time. The remaining terms can be computed
in \math{O(1)} time, because, for example,
\math{\e_i\transp((\scu+\delta_\scu)\matI-\matB)^{-2}\e_i
=(\scu+\delta_\scu-\s[i])^{-2}}.
The running time now drops to \math{O\left(rnk^2\right)}, since all the operations
on \math{\cl U} only contribute \math{O(rn)}.

\subsection{Proof of Correctness}

We prove that the output of Algorithm \ref{alg:2set} satisfies Lemma \ref{lemma:intro1}. Our proof is similar to the proof of Theorem 3.1~\cite{BSS09}. The main difference is that we need to accommodate two different sets of vectors. Let $\matW \in \R^{m \times m}$ be a positive semi-definite matrix with eigendecomposition
\mand{
\matW=\sum_{i=1}^{m}\lambda_i(\matW) \u_i\u_i\transp
}
and recall the functions $\phil(\scl,\matW)$, $\phiu(\scu,\matW)$, $L(\v,\delta_\scl,\matW,\scl)$, and $U(\v,\delta_\scu,\matW,\scu)$ defined in eqns.~(\ref{eqn:PD2}), (\ref{eqn:PD3}), and~(\ref{eqn:PD4}). We now quote two lemmas
proven  in~\cite{BSS09} using the Sherman-Morrison-Woodbury identity. These lemmas allow us to control the smallest and largest eigenvalues of $\matW$ under a rank-one perturbation.

\begin{lemma}
\label{lemma:sp1}
Fix \math{\delta_\scl>0}, \math{\matW \in \R^{m \times m}}, \math{\v \in \R^m}, and \math{\scl < \lambda_{m}(\matW)}. If \math{t > 0} satisfies
\mand{
t^{-1}\le L(\v,\delta_\scl,\matW,\scl),}
then \math{\lambda_{m}(\matW+t\v\v\transp) \ge \scl+\delta_\scl}.
\end{lemma}

\begin{lemma}
\label{lemma:sp2}
Fix \math{\delta_\scu>0}, \math{\matW \in R^{m \times m}}, \math{\v \in R^m}, and
\math{\scu > \lambda_{1}(\matW)}. If \math{t} satisfies
\mand{
t^{-1}\ge U(\v,\delta_\scu,\matW,\scu),}
then \math{\lambda_{1}(\matW+t\v\v\transp) \le\scu+\delta_\scu}.
\end{lemma}

Now recall that Algorithm~\ref{alg:2set} runs in \math{r} steps. Initially, all $n$ weights are set to zero. Assume that at the $\tau$-th step ($\tau = 0,\ldots,r-1$) the vector of weights 
$$\s_{\tau} =
\left[\s_{\tau}[1],\ldots, \s_{\tau}[n]\right]$$ has been constructed and let
$$\matA_\tau=\sum_{i=1}^n \s_{\tau}[i]\v_i\v_i\transp
\qquad \mbox{and} \qquad \matB_\tau=\sum_{i=1}^n\s_{\tau}[i]\u_i\u_i\transp.$$
Note that both matrices $\matA_\tau$ and $\matB_\tau$ are positive
semi-definite. We claim the following lemma
which guarantees that the algorithm is well-defined.
The proof is deferred to the next subsection.
\begin{lemma}
\label{lemma:feasible}
At the $\tau$-th step, for all $\tau=0,\ldots,r-1$, there exists an index $j$ in $\left\{1,\ldots,n\right\}$ such that setting the weight $t > 0$ as in Eqn.~(\ref{eqn:weight}) satisfies
\begin{equation}\label{eqn:PD5}
U(\u_j,\delta_\scu,\matB_\tau,\scu_\tau) \le t^{-1} \le L(\v_j,\delta_\scl,\matA_\tau,\scl_\tau).
\end{equation}
\end{lemma}

Once an index $j$ and a weight $t>0$ have been computed, Algorithm~\ref{alg:2set} updates the $j$-th weight in the vector of weights $\s_{\tau}$ to create the vector of weights $\s_{\tau+1}$. Clearly, at each of the $r$ steps, only one element of the vector of weights is updated. Since $\s_0$ is initialized to the all-zeros vector, after all $r$ steps are completed, at most $r$ weights are non-zero. The following lemma argues that $\lambda_{\min}(\matA_\tau)$ and $\lambda_{\max}(\matB_\tau)$ are bounded.
\begin{lemma}
\label{lemma:bounds}
At the $\tau$-th step, for all $\tau=0,\ldots,r-1$, $\lambda_{\min}(\matA_\tau)\ge\scl_\tau$ and
$\lambda_{\max}(\matB_\tau)\le \scu_\tau$.
\end{lemma}

\begin{proof}
Recall Eqn.~(\ref{eqn:PD1}) and observe that \math{\scl_0=-\sqrt{rk}<0} and \math{\scu_0=\delta_\scu\sqrt{r\ell}>0}. Thus, the lemma holds at \math{\tau=0}. It is also easy to verify that \math{\scl_{\tau+1}=\scl_{\tau}+\delta_\scl}, and, similarly, \math{\scu_{\tau+1}=\scu_{\tau}+\delta_\scu}. Now, at the $\tau$-step, given an index $j$ and a corresponding weight $t > 0$ satisfying
Eqn.~(\ref{eqn:PD5}), Lemmas \ref{lemma:sp1} and \ref{lemma:sp2} imply that
\eqan{\lambda_{\min}(\matA_{\tau+1})=
\lambda_{\min}(\matA_{\tau}+t\v_j\v_j\transp)\ge \scl_{\tau}+
\delta_{\scl}=\scl_{\tau+1};\\
\lambda_{\min}(\matB_{\tau+1}) = \lambda_{\max}(\matB_{\tau}+t\u_j\u_j\transp)\le \scu_{\tau}+
\delta_{\scu}=\scu_{\tau+1}.
}
The lemma now follows by simple induction on $\tau$.
\end{proof}

We are now ready to conclude the proof of Lemma~\ref{lemma:intro1}. By Lemma~\ref{lemma:bounds}, at the $r$-th step,
\mand{
\lambda_{\max}(\matB_r)\le
\scu_r \qquad\text{ and }\qquad
\lambda_{\min}(\matA_r)\ge\scl_r.
}
Recall the definitions of $\scu_r$ and $\scl_r$ from Eqn.~(\ref{eqn:PD1}) and note that they are both positive and well-defined because \math{r>k}. Lemma \ref{lemma:intro1} now follows after
rescaling the vector of weights $\s$ by \math{r^{-1}\left(1-\sqrt{k/r}\right)}. Note that the rescaling does not change the number of non-zero elements of \math{\s}, but does rescale all the eigenvalues of \math{\matA_r} and \math{\matB_r}.

\subsection{Proof of Lemma~\ref{lemma:feasible}}

In order to prove Lemma~\ref{lemma:feasible} we will use the following averaging argument.
\begin{lemma}\label{lemma:feasiblebound}
At any step $\tau=0,\ldots,r-1$,
\mand{
\sum_{i=1}^n U(\u_i,\delta_\scu,\matB_\tau,\scu_\tau)
\le
1-\sqrt{\frac{k}{r}}
\le
\sum_{i=1}^n L(\v_i,\delta_\scl,\matA_\tau,\scl_\tau).
}
\end{lemma}

\begin{proof}
For notational convenience, let \math{\phiu_\tau=\phiu(\scu_\tau,\matB_\tau)} and let \math{\phil_\tau=\phil(\scl_\tau,\matA_\tau)}. At \math{\tau=0}, \math{\matB_0=\bm0} and \math{\matA_0=\bm0} and thus \math{\phiu_0=\ell/\scu_0} and \math{\phil_0=-k/\scl_0}. Focus on the $\tau$-th step and assume that the algorithm has run correctly up to that point. Then,
\math{\phiu_\tau\le\phiu_0} and \math{\phil_\tau\le\phil_0}. Both are true at $\tau=0$ and, assuming that the algorithm has run correctly until the $\tau$-th step, Lemmas~\ref{lemma:sp1} and~\ref{lemma:sp2} guarantee that \math{\phiu_\tau} and \math{\phil_\tau} are non-increasing.

First, consider the upper bound on \math{U}. In the following derivation, $\lambda_i$ denotes the $i$-th eigenvalue of $\matB_{\tau}$. Using
\math{\trace(\u\transp\matX\u)=\trace(\matX\u\u\transp)} and
\math{\sum_{i}\u_i\u_i\transp=\matI_\ell}, we get
\eqan{
\sum_{i=1}^n U(\u_i,\delta_\scu,\matB_\tau,\scu_\tau)&=&
\frac
{\trace\left[(\scu_{\tau+1}\matI_{\ell}-\matB_\tau)^{-2}\right]}
{\phiu_\tau-\phiu(\scu_{\tau+1},\matB_\tau)}
+\phiu(\scu_{\tau+1},\matB_\tau)\\
&=&
\frac
{\sum_{i=1}^\ell\frac{1}{(\scu_{\tau+1}-\lambda_i)^2}}
{\delta_{\scu}\sum_{i=1}^\ell\frac{1}
{(\scu_{\tau+1}-\lambda_i)(\scu_{\tau}-\lambda_i)}}
+
\sum_{i=1}^\ell\frac{1}{(\scu_{\tau+1}-\lambda_i)}\\
&=&
\frac{1}{\delta_\scu}+\phiu_\tau
-\frac{1}{\delta_\scu}\left(1-
\frac
{\sum_{i=1}^\ell\frac{1}{(\scu_{\tau+1}-\lambda_i)^2}}
{\sum_{i=1}^\ell\frac{1}
{(\scu_{\tau+1}-\lambda_i)(\scu_{\tau}-\lambda_i)}}
\right)\\
&-&
\delta_\scu
\sum_{i=1}^\ell\frac{1}{(\scu_{\tau}-\lambda_i)(\scu_{\tau+1}-\lambda_i)}\\
&\le&
\frac{1}{\delta_\scu}+\phiu_0.
}
The last line follows because the last two
terms are negative (using the fact that \math{\scu_{\tau+1}>\scu_\tau>\lambda_i}) and \math{\phiu_\tau\le\phiu_0}. Now,
using \math{\phiu_0=\delta_\scu\sqrt{r\ell}} and the definition of
\math{\delta_\scu}, the upper bound follows:
\mand{
\frac{1}{\delta_\scu}+\phiu_0=\frac{1}{\delta_\scu}+
\frac{\ell}{\delta_\scu\sqrt{r\ell}}=\frac{1}{\delta_\scu}
\left(1+\sqrt{\frac{\ell}{r}}\right)=1-\sqrt{\frac{k}{r}}.}
In order to prove the lower bound on \math{L} we use a similar argument. Let $\lambda_i$ denote the $i$-th eigenvalue of $\matA_{\tau}$. Then,
\eqan{
\sum_{i=1}^n L(\v_i,\delta_\scl,\matA_\tau,\scl_\tau)&=&
\frac
{\trace\left[(\matA_\tau-\scl_{\tau+1}\matI_k)^{-2}\right]}
{\phil(\scl_{\tau+1},\matA_\tau)-\phil_\tau}
-\phil(\scl_{\tau+1},\matA_\tau)\\
&=&
\frac
{\sum_{i=1}^k\frac{1}{(\lambda_i-\scl_{\tau+1})^2}}
{\delta_{\scl}\sum_{i=1}^k\frac{1}
{(\lambda_i-\scl_{\tau+1})(\lambda_i-\scl_{\tau})}}
-
\sum_{i=1}^k\frac{1}{(\lambda_i-\scl_{\tau+1})}\\
&=&
\frac{1}{\delta_\scl}-\phil_\tau
+
\frac{1}{\delta_\scl}\left(
\frac
{\sum_{i=1}^k\frac{1}{(\lambda_i-\scl_{\tau+1})^2}}
{\sum_{i=1}^k\frac{1}
{(\lambda_i-\scl_{\tau+1})(\lambda_i-\scl_{\tau})}}
-1
\right)\\
&-&\delta_\scl
\sum_{i=1}^k\frac{1}{(\lambda_i-\scl_{\tau})(\lambda_i-\scl_{\tau+1})}\\
&\ge&
\frac{1}{\delta_\scl}-\phil_0 + \cl E.
}
Assuming \math{\cl E\ge 0} the claim follows immediately because \math{\delta_\scl=1} and 
$$\phil_0=-k/\scl_0=k/\sqrt{rk}=\sqrt{k/r}.$$ Thus, we only need to show that \math{\cl E\ge 0}. From the Cauchy-Schwarz inequality, for \math{a_i,b_i\ge 0},
\math{\left(\sum_i a_i b_i\right)^2\le\left(\sum_ia_i^2b_i\right)\left(\sum_i b_i\right)} and thus
\begin{eqnarray}
\nonumber \cl E\sum_{i=1}^k\frac{1}
{(\lambda_i-\scl_{\tau+1})(\lambda_i-\scl_{\tau})}
&=&\frac{1}{\delta_{\scl}}
\sum_{i=1}^k\frac{1}{(\lambda_i-\scl_{\tau+1})^2(\lambda_i-\scl_{\tau})}\\
&-&
\delta_\scl
\left(\sum_{i=1}^k\frac{1}{(\lambda_i-\scl_{\tau})
(\lambda_i-\scl_{\tau+1})}\right)^2\\
\nonumber &\ge&
\frac{1}{\delta_{\scl}}\sum_{i=1}^k\frac{1}{(\lambda_i-\scl_{\tau+1})^2(\lambda_i-\scl_{\tau})}\\
&-&
\delta_\scl
\sum_{i=1}^k\frac{1}{(\lambda_i-\scl_{\tau+1})^2(\lambda_i-\scl_{\tau})}
\sum_{i=1}^k\frac{1}{\lambda_i-\scl_\tau}\\
\label{eqn:PDD1} &=&
\left(\frac{1}{\delta_{\scl}}-\delta_\scl\phil_\tau\right)
\sum_{i=1}^k\frac{1}{(\lambda_i-\scl_{\tau+1})^2(\lambda_i-\scl_{\tau})}.
\end{eqnarray}
To conclude our proof, first note that $$\delta_\scl^{-1}-\delta_\scl\phil_\tau\ge \delta_\scl^{-1}-\delta_\scl\phil_0=1-\sqrt{k/r}>0$$
(recall that \math{r>k}). Second, \math{\lambda_i>\scl_{\tau+1}} because
\mand{
\lambda_{\min}(\matA_\tau)>\scl_\tau+\frac{1}{\phil_\tau}
\ge \scl_\tau+\frac{1}{\phil_0}=\scl_\tau+\sqrt{\frac{r}{k}}
>\scl_\tau+1=\scl_{\tau+1}.
}
Combining these two observations with Eqn.~(\ref{eqn:PDD1}) we conclude that $\cl E \geq 0$.
\end{proof}

Lemma~\ref{lemma:feasible} follows from Lemma~\ref{lemma:feasiblebound} because the two inequalities must hold simultaneously for at least one index \math{j}.

\section{Dual-set Spectral-Frobenius Sparsification: proof of Lemma~\ref{lemma:intro2}}
\label{sec:proofpure2}
In this section we will provide a constructive proof of Lemma~\ref{lemma:intro2}. Our proof closely follows the proof of Lemma \ref{lemma:intro1}, so we will only highlight the differences. We first discuss modifications to Algorithm~\ref{alg:2set}. First of all, the new inputs are \math{\cl V=\{\v_1,\ldots,\v_n\}} and \math{\cl A=\{\a_1,\ldots,\a_n\}}. The output is a set of $n$ non-negative weights $s_i$, at most $r$ of which are non-zero. We define the parameters
\mand{
\delta_\scl=1;
\qquad \delta_\scu=\frac{\sum_{i=1}^n \TNormS{\a_i}}{1-\sqrt{\frac{k}{r}}};
\qquad \scl_\tau=\tau-\sqrt{rk};
\qquad \scu_\tau=\tau \delta_\scu,
}
for all \math{\tau= 0,\ldots,r-1}. Let $\s_{\tau}$ denote the vector of weights at the $\tau$-th step of Algorithm~\ref{alg:2set} and initialize $\s_{0}$ and $\matA_0$ as in Algorithm~\ref{alg:2set} ($\matB_0$ will not be necessary). We now define the function $U_F\left(\a,\delta_\scu\right)$, where $\a \in \mathbb{R}^\ell$ and $\delta_\scu \in \mathbb{R}$:
\begin{equation}\label{eqn:UF}
U_F\left(\a,\delta_\scu\right) = \delta_\scu^{-1} \a\transp \a.
\end{equation}
Then, at the $\tau$-th step, the algorithm will pick an index $j$ and compute a weight $t>0$ such that
\mld{
U_F(\a_j,\delta_\scu)\le t^{-1}\le
L(\v_j,\delta_\scl,\matA_\tau,\scl_\tau).
\label{eq:updateF}
}
The algorithm updates the vector of weights $\s_{\tau}$ and the matrix
$$\matA_\tau=\sum_{i=1}^n s_{\tau,i}\v_i\v_i\transp.$$
It is worth noting that the algorithm does not need to update the matrix
$$\matB_\tau=\sum_{i=1}^ns_{\tau,i}\a_i\a_i\transp,$$
because the function $U_F$ does not need $\matB_{\tau}$ as input. To prove the correctness of the algorithm we need the following two intermediate lemmas.
\begin{lemma}\label{lemma:feasibleF}
At every step \math{\tau=0,\ldots,r-1} there exists an index $j$ in $\left\{1,\ldots,n\right\}$ that satisfies Eqn.~(\ref{eq:updateF}).
\end{lemma}
\begin{proof}
The proof is very similar to the proof of Lemma~\ref{lemma:feasible} (via Lemma~\ref{lemma:feasiblebound}) so we only sketch the differences. First,
note that the dynamics of \math{L} have not been changed and thus the lower bound for the average of $L(\v_j,\delta_\scl,\matA_\tau,\scl_\tau)$ still holds.
We only need to upper bound the average of \math{U_F(\a_i,\delta_\scu)} as in Lemma~\ref{lemma:feasiblebound}. Indeed,
\mand{
\sum_{i=1}^n U_F(\a_i,\delta_\scu)
=\delta_\scu^{-1}\sum_{i=1}^n\a_i\transp\a_i
=
\delta_\scu^{-1}\sum_{i=1}^n\TNormS{\a_i}
= 1-\sqrt{\frac{k}{r}},
}
where the last equality follows from the definition of \math{\delta_\scu}.
\end{proof}

\begin{lemma}\label{lemma:sp2F}
Let $W \in \mathbb{R}^{\ell \times \ell}$ be a symmetric positive semi-definite matrix, let $\a \in \mathbb{R}^{\ell}$ be a vector, and let
\math{\scu \in \mathbb{R}} satisfy $\scu > \trace(\matW)$. If \math{t > 0} satisfies
\mand{U_F\left(\a,\delta_\scu\right) \le
t^{-1},
}
then $$\trace\left(\matW+t\v\v\transp\right)\le\scu+\delta_\scu.$$
\end{lemma}

\begin{proof}
Using the conditions of the lemma and the definition of $U_F$ from Eqn.~(\ref{eqn:UF}),
\eqan{
\trace(\matW+t\a\a\transp)-\scu-\delta_\scu,
&=&
\trace(\matW)-\scu+t\a\transp\a-\delta_\scu,\\
&\le&
\trace(\matW)-\scu < 0,
}
which concludes the proof of the lemma.
\end{proof}

We can now combine Lemmas~\ref{lemma:sp1} and \ref{lemma:sp2F} to prove that at all steps $\tau=0,\ldots,r-1$,
$$\lambda_{\min}(\matA_\tau)\ge\scl_\tau \qquad \mbox{and}\qquad
\trace(\matB_\tau)\le \scu_\tau.
$$
Note that after all $r$ steps of the algorithm are completed, $$\scl_r=r\left(1-\sqrt{k/r}\right),$$
and
$$\scu_r=r\left(1-\sqrt{k/r}\right)^{-1}\sum_{i=1}^n\TNormS{\a_i}.$$
A simple rescaling now concludes the proof.
The running time of the (modified) Algorithm~\ref{alg:2set} is $O\left(nrk^2+n\ell \right)$, where the latter term emerges from the need to compute the function $U_F(\a_j,\delta_\scu)$ for all $j=1,\ldots,n$ once throughout the algorithm.


\section{Lower bounds}

\subsection{Spectral Norm Approximation}\label{sec:lower}
Theorem~\ref{theorem:lower1} below is the main result in this section.
\begin{theorem}
\label{theorem:lower1}
For any \math{\alpha>0}, any $k \geq 1$, and any $r \geq 1$, there exists a matrix \math{\matA \in \mathbb{R}^{(n+1) \times n}} for which
$$\frac{\TNorm{\matA-\matC\matC^+\matA}^2}{\TNorm{\matA-\matA_k}^2} \ge\frac{n+\alpha^2}{r+\alpha^2}.$$
Here $\matC$ is any matrix that consists of $r$ columns of $\matA$. As $\alpha \rightarrow 0$, the lower bound is $n/r$ for the approximation ratio of spectral norm column-based matrix reconstruction.
\end{theorem}

\begin{proof}
We extend the lower bound in \cite{DR10} to arbitrary $r > k$. Consider the matrix
$$\matA = [\e_1+\alpha\e_2, \e_1+\alpha\e_3,\ldots, \e_1+\alpha\e_{n+1}] \in\R^{(n+1)\times n},$$
where \math{\e_i\in\R^{n+1}} are the standard basis vectors. Then,
$$\matA\transp\matA=\bm{1}_n\bm{1}_n\transp+\alpha^2\matI_{n}, \qquad
\sigma_1^2(\matA)=n+\alpha^2, \qquad \mbox{and}  \qquad
\sigma_i^2(\matA)=\alpha^2 \mbox{\ \ for\ \ } i>1.$$
Thus, for all $k\ge1$, $\TNorm{\matA-\matA_k}^2=\alpha^2$. Intuitively, as \math{\alpha\rightarrow0}, \math{\matA} is a rank-one matrix. Consider any \math{r} columns of $\matA$ and note that, up to row permutations, all sets of $r$ columns of $\matA$ are equivalent. So, without loss of generality, let \math{\matC} consist of the first \math{r} columns of \math{\matA}. We now compute the optimal reconstruction of $\matA$ from $\matC$ as follows: let $\a_j$ be the $j$-th column of $\matA$. In order to reconstruct \math{\a_j}, we minimize \math{\TNormS{\a_j-\matC\x}} over all vectors \math{\x\in\R^{r}}. Note that if \math{j\le r} then the reconstruction error is zero. For $j > r$, \math{\a_j=\e_1+\alpha\e_{j+1}}, 
$$\matC\x=\e_1\sum_{i=1}^r x_i+\alpha\sum_{i=1}^rx_i\e_{i+1}.$$ Then,
\eqan{
\TNorm{\a_j-\matC\x}^2
&=&
\TNorm{\e_1\left(\sum_{i=1}^rx_i-1\right)+
\alpha\sum_{i=1}^rx_i\e_{i+1}-e_{j+1}}^2\\
&=&
\left(\sum_{i=1}^rx_i-1\right)^2+\alpha^2\sum_{i=1}^rx_i^2+1.
}
The above quadratic form in \math{\x} is minimized when \math{x_i=\left(r+\alpha^2\right)^{-1}} for all $i=1,\ldots,r$. Let $\hat\matA = \matA - \matC\matC^+
\matA$ and let the $j$-th column of $\hat\matA$ be $\hat\a_j$. Then, for \math{j\le r}, \math{\hat\a_j} is an all-zeros vector; for \math{j>r}, \math{\hat\a_j=\alpha\e_{j+1}-\frac{\alpha}{r+\alpha^2}\sum_{i=1}^r\e_{i+1}}.
Thus,
$$\hat\matA\transp\hat\matA=
\left[
\begin{matrix}
\bm0_{r\times r}&\bm0_{r\times (n-r)}\\
\bm0_{(n-r)\times r}&\matZ
\end{matrix}
\right],$$
where $$\matZ=\frac{\alpha^2}{r+\alpha^2}\bm1_{n-r}\bm1_{n-r}\transp +\alpha^2\matI_{n-r}.$$
This immediately implies that
\eqan{
\TNorm{\matA-\matC\matC^+\matA}^2
&=&
\TNorm{\hat\matA}^2 = \TNorm{\hat\matA\transp \hat\matA} = \TNorm{\matZ}^2= \frac{(n-r)\alpha^2}{r+\alpha^2}+\alpha^2
=   \frac{n+\alpha^2}{r+\alpha^2}\alpha^2.
}
This concludes our proof, because $\alpha^2=\TNormS{\matA-\matA_k}.$
\end{proof}

\subsection{Frobenius norm approximation}\label{sec:lowerF}
\noindent Note that a lower bound for the ratio $$\XNormS{\matA-\Pi_{\matC,k}^{\xi}(\matA)}/\XNormS{\matA-\matA_k},$$
does not imply a lower bound for the ratio
$$\XNormS{\matA-\matC\matC^+\matA}/\XNormS{\matA-\matA_k},$$
because
$$ \XNormS{\matA-\matC\matC^+\matA}/\XNormS{\matA-\matA_k} \le \XNormS{\matA-\Pi_{\matC,k}^{\xi}(\matA)}/\XNormS{\matA-\matA_k}. $$
Also, note that Proposition 4 in~\cite{DV06}
shows a lower bound equal to $\left(1 + k/2r\right)$ for the ratio $$\FNormS{\matA-\Pi_{\matC,k}^{\mathrm{F}}(\matA)}/\FNormS{\matA-\matA_k}.$$
For completeness, we extend the bound of~\cite{DV06} to the ratio $$\FNormS{\matA-\matC\matC^+\matA}/\FNormS{\matA-\matA_k}.$$
In fact, we obtain a lower bound which is asymptotically $1 + k/r$.
We start with the following lemma.
\begin{lemma}\label{lemma:lowerF}
For any \math{\alpha>0} and \math{r\ge 1}, there
exists a matrix \math{\matA\in\R^{m\times n}} for which
\mand{
\frac{\FNormS{\matA-\matC\matC^+\matA}}{\FNormS{\matA-\matA_1}}
\ge \frac{n-r}{n-1}\left(1+\frac{1}{r+\alpha^2}\right).
}
\end{lemma}
\begin{proof}
We use the same matrix $\matA\in\R^{(n+1)\times n}$ from
Theorem~\ref{theorem:lower1}. So, it follows that
$\FNormS{\matA-\matC\matC^+\matA}=
\trace(\matZ)=\alpha^2(n-r)\left(1+\frac{1}{r+\alpha^2}\right),$
and $\FNormS{\matA-\matA_1}=(n-1)\alpha^2$, which concludes the proof.
\end{proof}

\begin{theorem}\label{theorem:lowerF}
For any \math{\alpha>0}, any $k \geq 1$, and any $r \geq 1$, there exists a matrix \math{\matB \in \mathbb{R}^{m \times \ell}} for which
$$\frac{\FNorm{\matB-\matC\matC^+\matB}^2}{\FNorm{\matB-\matB_k}^2}
\ge\frac{\ell-r}{\ell-k}\left(1+\frac{k}{r+\alpha^2}\right).
$$
Here $\matC$ is any matrix that consists of $r$ columns of $\matB$.
By taking $\alpha \rightarrow 0$ and \math{\ell\rightarrow\infty} the  lower bound is $1+\left(k/r\right)$ for the approximation ratio of Frobenius
norm column-based matrix reconstruction.
\end{theorem}%
In our construction of \math{\matB},
\math{m=(n+1)k} and \math{\ell=nk} for any \math{\ell \ge r}.
\begin{proof}
The matrix $\matB$ is constructed as follows. Let $\matA$ have
dimensions $(n+1) \times n$ and be constructed as in
Theorem~\ref{theorem:lower1} except with \math{\alpha} replaced by
\math{\alpha'=\alpha/\sqrt{k}}.
$\matB$ is block diagonal, with $k$ copies of $\matA$ along its diagonal;
so,
the dimensions of \math{\matB} are \math{m=(n+1)k} and \math{\ell=nk}.
We sample \math{r} columns in total, with \math{r_i} columns from each block.
Lemma~\ref{lemma:lowerF} holds for each block, with \math{r}
replaced by \math{r_i}, where \math{r_i\ge 0} and
\math{\sum_{i=1}^kr_i=r}. We analyze the
Frobenius norm error in each block independently.
Let \math{\matZ_i} be the error matrix in each block, as in the proof of
Theorem~\ref{theorem:lower1}.
Then, using Lemma~\ref{lemma:lowerF},
the approximation error is equal to
\eqan{
\FNormS{\matB-\matC\matC^+\matB}&=&
\sum_{i=1}^k\trace(\matZ_i)
={\alpha'}^2\sum_{i=1}^k\left(n-r_i
\right)\left(1+\frac{1}{r_i+{\alpha'}^2}\right).
}
This last expression is minimized subject to the constraint that
\math{\sum_{i=1}^kr_i=r} when \math{r_i=r/k}, and so
\eqan{
\FNormS{\matB-\matC\matC^+\matB}&\ge&
{\alpha'}^2(nk-r)
\left(1+\frac{k}{r+k{\alpha'}^2}\right)
=
{\alpha'}^2(\ell-r)
\left(1+\frac{k}{r+{\alpha}^2}\right).
}
Where we used \math{\alpha^2=k{\alpha'}^2}.
 The result follows because $$\FNormS{\matB-\matB_k}=k\FNormS{\matA-\matA_1}=k(n-1){\alpha}'^2
=(\ell-k){\alpha'}^2.$$
\end{proof}


\section{Conclusions and Open Problems} \label{sec:open}

Several interesting questions remain unanswered, which we summarize in this section.

First, is it possible to improve the running time of the deterministic algorithms of Lemmas~\ref{lemma:intro1} and~\ref{lemma:intro2}? Recently, Zouzias~\cite{Zouzias} made progress in improving the running time of the spectral sparsification result of~\cite{BSS09}; can we get a similar improvement for the 2-set algorithms presented here? Or perhaps, can we trade off the running time with randomization in those algorithms?

Second, in the parlance of Theorem~\ref{thmFast3}, is there a \textit{deterministic} algorithm that selects \math{O(k/\epsilon)} columns from $\matA$ and guarantees relative-error accuracy for the error \math{\FNormS{ \matA - \Pi_{\matC,k}^{\mathrm{F}}(\matA)}}? Such a deterministic bound would be possible, for example, by derandomizing the adaptive sampling technique used in Theorem~\ref{thmFast3}. In a recent development, \cite{GS2011} partially answers this question by extending the volume sampling approach of~\cite{DR10} to deterministically select
\math{\frac{k}{\epsilon}(1+o(1))} columns
and obtain a relative error bound for the term \math{\FNormS{ \matA - \matC\matC^+\matA}}. Notice that
it is not obvious if~\cite{GS2011} implies a similar deterministic bound for the error
\math{\FNormS{ \matA - \Pi_{\matC,k}^{\mathrm{F}}(\matA)} }.

Third, is it possible to develop a (deterministic or randomized) algorithm to compute the best (in the \emph{spectral} norm) rank restricted 
projection of $\matA$, into a specific subspace
\math{\matC} efficiently. That is, can one
compute
$ \Pi_{\matC,k}^{\mathrm{2}}(\matA)$ in 
time that is \math{O(T_{SVD}(\matC))}?

Finally, is it possible to develop column reconstruction algorithms that provide upper bounds in terms of the column reconstruction that is achieved with an optimal set of columns?
For example, given $\matA$, $k$, and $r \ge k,$ let $\matC_{opt}$ be the best choice of $r$ columns in $\matA$. Is it possible to develop algorithms that
select matrices $\matC$ with $r$ columns and provide upper bounds of the following form,
$$\FNormS{ \matA - \matC\matC^+\matA} \le \gamma \cdot \FNormS{ \matA - \matC_{opt}\matC_{opt}^+\matA}, $$
where $\gamma \ge 1$ is the approximation factor. Also, how do the lower bounds change if we replace $\matA-\matA_k$ with $ \matA - \matC_{opt}\matC_{opt}^+\matA$?

\section*{Acknowledgments}
We would like to thank A. Deshpande, D. Feldman, K. Varadarajan, and J. Tropp for useful discussions, and D. Feldman and K. Varadarajan for pointing out the connections between the subspace approximation line of research~\cite{DV07,FL11, FMSW10,SV07} and our work. We would also like to thank an anonymous reviewer for 
pointing out \cite{Szy06} and how this result on oblique projections
implies Eqn.~(\ref{eqn2}) and its implications to Theorems~\ref{thmFast1} 
and~\ref{theorem:spectralIn}; this avenue 
suggested by the reviewer slightly improved our previous bounds.

This work has been supported by NSF CCF-1016501 and NSF DMS-1008983 to P. Drineas and M. Magdon-Ismail. C. Boutsidis acknowledges the support from XDATA program of the Defense Advanced Research Projects Agency (DARPA), administered through Air Force Research Laboratory contract FA8750-12-C-0323.

\end{document}